\documentclass[conference]{IEEEtran}
\usepackage{enumitem}
\usepackage{hyperref}
\hypersetup{
    colorlinks=true,
    linkcolor=black,
    filecolor=black, 
    citecolor=black,
    urlcolor=cyan,
}

\hypersetup{
    colorlinks=true,
    linkcolor=blue,
    filecolor=magenta,      
    urlcolor=cyan,
    pdfpagemode=FullScreen
    }
\usepackage{tikz}
\usepackage{amsmath}
\newtheorem{definition}{Definition}
\newtheorem{theorem}{Theorem}
\newtheorem{assumption}{Assumption}
\usepackage{filecontents}
\usepackage{verbatim} 
\usepackage{graphicx}
\usepackage{enumitem}
\usepackage{colortbl} 
\usepackage{xspace}
\usepackage{amsmath}
\usepackage{url}
\usepackage[framemethod=TikZ]{mdframed}
\usepackage[compatibility=false]{caption}  
\usepackage{subcaption}
\usepackage{float}
\usepackage{balance}
\usepackage{multirow} 
\usepackage{booktabs} 
\usepackage{xcolor}  
\usepackage{amsmath}
\usepackage{amssymb}
\usepackage{threeparttable}
\usepackage{soul}
\usepackage{circledsteps}
\newcommand{\name}{\texttt{PRIME}\xspace}

\newcommand{\mypara}[1]{\noindent\textbf{#1}}

\begin{document}
\date{}

\title{From Multiplicity to Vulnerability: Privacy Amplification Risk from One-Dataset-Multiple-Model Exposure}
% \title{Theoretical and Empirical Privacy Amplification Risk from One-Dataset-Multiple-Model Exposure}

%for single author (just remove % characters)
\author{
\IEEEauthorblockN{
Qirui Huang\textsuperscript{1},
Na Li\textsuperscript{2},
Hongsheng Hu\textsuperscript{3},
Zhi Zhang\textsuperscript{1},
Anmin Fu\textsuperscript{2},
Yansong Gao\textsuperscript{1}
}

\IEEEauthorblockA{\textsuperscript{1}The University of Western Australia, Australia}

\IEEEauthorblockA{\textsuperscript{2}Nanjing University of Science and Technology, China}

\IEEEauthorblockA{\textsuperscript{3}The University of Newcastle, Australia}
}

\IEEEoverridecommandlockouts
\makeatletter\def\@IEEEpubidpullup{6.5\baselineskip}\makeatother

\maketitle

%-------------------------------------------------------------------------------
\begin{abstract}
To efficiently exploit a valuable data source (e.g., facial or medical images), it is frequently harnessed to fulfill multiple learning objectives (e.g., facial recognition, age estimation, and race classification). Each trained model is then deployed as an independent API service for corresponding inference. However, the privacy risk introduced by this one-dataset-multiple-model (ODMM) paradigm is completely overlooked by the community.

For the first time, this work reveals that the ODMM setting substantially amplifies privacy leakage. 
We establish a theoretical framework that proves that privacy leakage accumulates as more ODMM models are exposed, a phenomenon we term ODMM privacy composition. Guided by this theoretical foundation, we propose \name (\textbf{P}rivacy Amplification \textbf{RI}sk from One-Dataset-Multiple-\textbf{M}odel \textbf{E}xposure) to systematically assess this risk and quantify the resulting leakage using membership inference attacks (MIAs). Under black-box access to ODMM models, we design an aggregation mechanism that collectively captures carefully identified privacy signals leaked by individual ODMM models, and construct an attack meta-classifier over the aggregated meta-information to infer the membership status of a given sample jointly.
Our results provide strong evidence that dataset reuse across ODMM models strikingly jeopardizes privacy,
which is consistently evident across five privacy-sensitive image and textual benchmark datasets and diverse model architectures (from ResNet and ViT to Qwen3-1.7B), spanning three domains: facial analysis, medical imaging, and textual attribution analysis.
While mitigations such as differential privacy can reduce the effectiveness of \name with trade-offs, our attack still consistently outperforms single-task MIAs. 
% Overall, our findings underscore that privacy considerations must be revisited when reusing data to achieve cost-efficient deployment of ODMM models.

\end{abstract}
\IEEEpeerreviewmaketitle

%-------------------------------------------------------------------------------
% ======================================================
\section{Introduction}
\label{sec:intro}
% ======================================================
The rapid advancement of deep learning (DL) has profoundly reshaped diverse domains such as autonomous driving~\cite{hu2023planning}, facial recognition~\cite{deng2019arcface}, and medical diagnosis~\cite{litjens2017survey}, while the recent emergence of Artificial Intelligence Generated Content (AIGC), ranging from text generation (e.g., ChatGPT~\cite{chatgpt}) to image synthesis (e.g., Midjourney~\cite{midjourney}), has further propelled this transformative wave. 
%However, this remarkable progress is fundamentally underpinned by the availability of massive high-dimensional datasets tailored to different tasks, which remain practically challenging to acquire due to prohibitive collection costs and inherent data scarcity.
For this success, high-quality data has been recognized as a critical factor for improving model performance in addition to scaled model size~\cite{wang2024greats,li2024quantity}. However, this remarkable progress is fundamentally underpinned by the availability of massive task-specific datasets, which remain practically challenging to acquire due to prohibitive collection costs and inherent data scarcity. 

Fortunately, a given dataset typically harbors rich information suitable for diverse learning objectives~\cite{vandenhende2021multi, zamir2018taskonomy,ye2022inverted}. 
Consequently, to overcome data scarcity and maximize valuable resource utility, it has become standard practice across both academia and industry to train multiple individual models over the same dataset (e.g., CelebA), with each addressing a specific task (e.g., age estimation, race, gender, and facial expression recognition). 
For another instance, Facebook AI Research's Detectron2~\cite{detectron2} releases multiple independent task-specific vision models in its model zoo, including object detection, instance segmentation, human keypoint estimation, and panoptic segmentation, all trained on the COCO dataset~\cite{lin2014microsoft}.
%For instance, the benchmark CelebA dataset~\cite{liu2015deep}, originally curated for facial attribute recognition, has been extensively repurposed---across open-source model zoo, i.e., Hugging Face---for tasks including gender classification~\cite{celebagender}, smile detection~\cite{celebasmile}, and image orientation detection~\cite{celebaorientation}.

% \vspace{0.1cm}
\noindent\textbf{Overlooked Privacy Problem:}
Although repurposing the same data source for multiple tasks optimizes resource efficiency, the \textit{privacy risks} posed by the collective exposure of \textit{one-dataset-multiple-model} (ODMM) are essentially overlooked and poorly understood, despite the privacy vulnerabilities of \textit{individual} DL model has been widely studied~\cite{li2022auditing,baluta2022membership,ye2022enhanced,li2021membership,shokri2017membership,carlini2022membership,li2024seqmia,shang2025defend}.

MIAs serve as a critical tool for auditing such privacy leakage~\cite{WangZCBKY25, yeom2018privacy, chen2020gan, shokri2017membership,li25enhanced,chen2025method}. By exploiting the DL model’s memorization of training data, MIAs distinguish subtle differences in model outputs between members and non-members, thus enabling attackers to infer whether a specific sample was included in the training set. For example, determining an individual’s membership in a facial recognition training dataset for specialized psychiatric cohorts could inadvertently disclose sensitive mental health information, thereby compromising personal privacy.
Building on this, from another perspective, MIAs can be utilized to audit the risks of user privacy leakage and assess deployed model compliance with privacy regulations such as the General Data Protection Regulation (GDPR)~\cite{gdpr}.

Notably, each task-specific model learns distinct representations from the same underlying data distribution and is independently optimized toward its unique objective, which may lead to different memorization of the shared training samples, thereby raising potentially amplified privacy leakage. However, when multiple models upon one dataset, later referred to as ODMM models, are released and queried for inference, the privacy risks stemming from their divergent memorization behaviors remain \textit{overlooked}. Hence, we pose the following research questions to highlight the urgent need for a comprehensive investigation into the completely overlooked privacy risks from the exposure of ODMM models.

\begin{mdframed}[backgroundcolor=black!10,rightline=false,leftline=false,topline=false,bottomline=false,roundcorner=2mm]
Does the collective exposure of ODMM models derived exacerbate privacy leakage of their shared underlying dataset?
If so, to what degree does it amplify privacy risks?
\end{mdframed}

\noindent\textbf{Our Work:}
This work, for the first time, unveils and proves that access to multi-task models trained on the same underlying dataset exacerbates privacy leakage. 
Firstly, we theoretically prove that privacy leakage accumulates as more ODMM models are exposed. Generally, our proof is inspired by the composition theorem in differential privacy (DP), which formally establishes that privacy loss accumulates across multiple queries, degrading the protection afforded to any individual~\cite{dwork2006calibrating,dwork2014algorithmic,ganta2008composition,dinur2003revealing}. We apply this concept to this neglected new ODMM context and lay out a theoretical foundation as \textit{ODMM privacy composition}.

% This gap is particularly concerning given that the composition theorem in differential privacy (DP) formally establishes that privacy loss accumulates across multiple disclosures, degrading the protection afforded to any individual~\cite{dwork2006calibrating,dwork2014algorithmic,ganta2008composition,dinur2003revealing}. Nonetheless, the privacy amplification risk in the ODMM setting, where multiple non-DP models are independently trained on the same dataset and deployed as black-box APIs, has not been systematically investigated.
Guided by the theoretical foundation, through the lens of MIA as a privacy auditing tool, we propose and design \name as a framework for systematically and empirically assessing the degree of additional leakage across various factors and tasks.
%The primary reason is that each task-specific models leak privacy in slightly different ways due to divergent memorization patterns for the same sample.
% Fundamentally, this vulnerability stems from the divergent memorization patterns inherent to distinct tasks, where each ODMM model can expose unique information about the same sample. 
Aggregating accumulated information exposed by ODMM models amplifies the overall privacy risk, even in the focused \textit{black-box} MIA. Below, we highlight the key findings of \name and briefly describe its core design.
%\garrison{The rest also requires update. We do not need to highlight empirical insight anymore or use some terms like divergent memorization patterns, complementary memorization. We can consistently name it as `ODMM privacy composition'.}

\noindent$\bullet$\textit{Theoretical Foundation.} We formalize ODMM privacy composition as a unified theoretical framework, demonstrating that collectively exposing ODMM models fundamentally amplifies privacy risk beyond the leakage from each individual model. First, we prove that the optimal membership inference advantage is monotonically non-decreasing as more models are exposed (Theorem~\ref{thm:monotone}), and that leakage accumulates additively across ODMM models with approximately independent per-model leakage signals (Theorem~\ref{thm:llr-decomp}). 
% Empirically, MIA performance improves monotonically, with AUC increasing from 65.6\% to 78.5\% as the number of exposed models grows from 3 to 7. 
However, tasks among ODMM models often have dependencies. Second, we theoretically characterize how such inter-task correlation affects the extent of privacy amplification, revealing that more heterogeneous or divergent (less correlated) tasks lead to more leakage than homogeneous ones (Theorem~\ref{thm:keff}). These theorems are supported by corresponding experimental evidence.
\noindent$\bullet$\textit{\name Framework.}
The above theorem lays out the foundations of the \name framework for systematic and empirical evaluation of ODMM privacy composition. 
Overall, a \textit{black-box adversary} leverages \textit{a given MIA feature extraction method} by querying each ODMM model to obtain corresponding raw outputs (e.g., posteriors or predicted values), which are then constructively transformed into (unique) task-specific feature vectors. Then, these vectors are strategically concatenated through a \textit{joint information fusion} step to form MIA meta-data, which is utilized to train a binary meta-classifier for membership inference.

We conduct comprehensive experiments on five image and textual benchmark datasets and five diverse model architectures
(from ResNet to ViT and Qwen3-1.7B) across three domains, including facial analysis, medical imaging, and large language models (LLMs) instruction fine-tuning to comprehensively evaluate the effectiveness of \name. Results affirm that \name consistently and \textit{significantly} outperforms all/maximum single-task/model in MIA performance, confirming that ODMM models trained on the same data source indeed substantially jeopardize privacy.
For instance, \name attains 86.9\% attack accuracy on UTKFace, surpassing the best single-task baseline by 18.4\%, with consistent improvements on HAM10000 (+7.8\%) and Blog Authorship Corpus (+9.8\%).
%For example, when evaluating on UTKFace with three tasks, \name achieves an AUC of 92.5\%, outperforming the average single-task attacks by over 25\%, while on HAM10000 (medical imaging) and Qwen3-1.7B fine-tuning (LLM domain), PRIME achieves 27.3\% and 19.7\% relative AUC improvements, respectively.
We also conduct ablation studies to analyze the effect of various factors, i.e., the number of exposed models 
%and the training data overlap rate, 
on the attack performance.
Finally, we provide further comparison and discussion on the privacy leakage of multi-task shared-backbone models compared to independently trained ODMM models, as well as potential defense mechanisms. Our main contributions are:

%\noindent$\bullet$ We discover and analyze the phenomenon that ODMM models memorize the same data differently, raising great privacy concerns. To the best of our knowledge, this has been completely overlooked by the community.

\noindent$\bullet$ We establish the theoretical foundation that privacy leakage accumulates across ODMM models, and the degree of amplification can be characterized by inter-task correlations.

\noindent$\bullet$  We propose \name, the first framework for systematically assessing empirical evidence on the ODMM privacy composition through the lens of MIA, showing that such ODMM co-exposure indeed significantly jeopardizes privacy.

\noindent$\bullet$ We conduct extensive experiments across classic image domains (including privacy-critical clinical AI) and the LLM field to demonstrate the effectiveness and validation of \name, which is also independent of MIA methods such as hard-label-only~\cite {li2021membership} or recent RMIA attacks~\cite {ZarifzadehLS24}. 
% The evaluated ODMM model task is diverse, from binary classification and multiple-class classification to regression and generative tasks.
% ======================================================
\section{Background and Related Work}
\label{sec:related}
% ======================================================

% \mypara{Membership Inference Attacks.} 
MIAs have been widely used for evaluating the privacy risks of DL models by inferring whether a particular sample belongs to the target model's training dataset~\cite{li2021membership,carlini2022membership,Wang2025rigging,ye2022enhanced,shang2025defend, HuiYYBGC21,zhang25soft}.
Formally, given a target sample $x$ and the victim model $\mathcal{M}$, the inference algorithm $\mathcal{A}$ can be expressed as:
\begin{equation}
\mathcal{A}: x, \mathcal{M} \rightarrow \{ 0, 1 \}.
\end{equation}
If a target sample $x$ was used to train $\mathcal{M}$, $\mathcal{A}$ outputs 1 (i.e., member) and 0 otherwise (i.e., non-member).

In practice, most MIAs operate under a black-box setting, where the adversary is restricted to observing only the posterior probabilities (i.e., confidence scores) of the victim model, and can be broadly categorized into training-based and metric-based methods, depending on whether a dedicated attack meta-classifier is employed.
Training-based MIAs first train shadow models to mimic the victim model’s behavior, then use the outputs of these shadow models to construct labeled meta-data, which is subsequently used to train a binary attack meta-classifier for membership inference.
For example, Shokri \textit{et al.}~\cite{shokri2017membership} use posteriors, i.e.,  prediction confidences, as the meta-data. Yuan \textit{et al.}~\cite{yuan2022membership} integrate posteriors, sensitivity, and labels to train a transformer-based attack meta-classifier. 
In contrast, several studies~\cite{song2021systematic,yeom2018privacy,Salem0HBF019,ZarifzadehLS24,carlini2022membership,du2026cascading} introduced metric-based MIAs that leverage statistical measures calculated from the victim model, e.g., posterior~\cite{yeom2018privacy}, scaled confidence score~\cite{du2025imitative}, modified entropy~\cite{song2021systematic}, and likelihood ratio~\cite{carlini2022membership}, to distinguish membership status without the need to train the attack meta-classifier.
More recently, Du \textit{et al.}~\cite{du2026cascading} utilize membership dependencies between instances to enhance attack performance.

Moreover, some studies utilize MIA as a tool to investigate the privacy effect of specific systems or technologies, e.g., query-based systems~\cite{stevanoski2024querycheetah}, RAG datastore~\cite{NasehPSCOH25, GaoMD0G25}, machine unlearning~\cite{chen2021machine}, model compression~\cite{li2025compleak}, speaker recognition systems~\cite{chen2024slmia}, explainable machine learning methods~\cite{liu2024please}, and visual encoders~\cite{zhu2024unified,liu2021encodermi}. 
Recently, MIAs have been extended to the foundation models, e.g., LLMs~\cite{meeus2024did,wen2024membership,he2025llms,tong2026membership,chen2026window}, diffusion models~\cite{Peng2025diffence,pang2025black,wang2026inference}, vision-language models~\cite{hu2025vlms}, and text-to-video models~\cite{wang2026vidleaks}. 
However, several state-of-the-art studies~\cite{duan2024do,hayesexploring} debate that MIAs on large-scale generative foundation models typically perform at near-random levels. This is largely attributed to the substantially reduced overfitting resulting from massive training corpora combined with limited training iterations.

To the best of our knowledge, no prior work has explored membership leakage in scenarios where ODMM models are independently trained on the same dataset without mutual interference. The most relevant work from Yan \textit{et al.}~\cite{yan2023mtl} investigates MIA in the multi-task learning setting, where multiple tasks share an identical backbone and employ task-specific heads, which substantially differs from our work.That is, our work involves multiple models trained on the same dataset, whereas Yan \textit{et al.}~\cite{yan2023mtl} train only a single model on the dataset. In addition, one shortcoming of using the multi-task model is that its utility for a specific task is notably inferior, as the underlying backbone is not optimal for such a task. From the privacy perspective, their investigation differs fundamentally from ours, where new exploration could be derived by correlating information from independently trained ODMMs on the same dataset, thus leaking substantially more privacy than that of a multi-task network (utility and privacy comparisons are in Section~\ref{sec:discussion_mtl}).

\section{Theoretical Foundation: ODMM Privacy Composition}
\label{sec:theory}

We posit that the collective exposure of ODMM models fundamentally amplifies the privacy risk compared to the leakage arising from each model. We note that this claim conceptually aligns with the composition theorem in differential privacy (DP)~\cite{dwork2006calibrating,dwork2014algorithmic,ganta2008composition,dinur2003revealing}.
Specifically, under $\varepsilon$-DP, privacy loss $\varepsilon$ measures the maximum change in output distributions induced by the inclusion or exclusion of a single record. The composition of $k$ independent query-induced-release on the same underlying dataset, where each release satisfies $\varepsilon_i$-DP, incurs a total privacy loss bounded by $\sum_{i=1}^{k}\varepsilon_i$, implying a monotonic degradation of the privacy guarantee with each release. However, directly applying DP composition theorems to our setting is not straightforward, as the ODMM models are standard DL models that have nothing to do with the DP mechanism. Furthermore, we study empirical membership leakage under black-box access rather than the worst-case guarantees of formally private mechanisms. 

To fill this theoretical gap, we establish a novel theoretical framework for \textit{ODMM privacy composition} to formally analyze the amplification of privacy risks under the black-box MIA setting.
Next, we formalize ODMM as a joint observation channel and prove the monotonicity of the optimal attack advantage in the number of exposed ODMM models. We then derive an additive accumulation of privacy leakage under conditional independence, and further analyze how inter-task correlations affect the extent of privacy risk amplification.

\subsection{Joint Observation Channel of ODMM}
We begin by modeling the outputs of ODMM models as a multi-channel information release (Definition~\ref{def:odmm}), and formulate membership inference (MI) as a hypothesis testing over the joint information across channels (Definition~\ref{def:mia-game}).

\begin{definition}[Joint Multi-channel Observations]
\label{def:odmm}
    Let $\mathcal{D}$ be the shared training set and 
    let $z=(x,y)$ be a queried sample. Let 
    $f_1,\ldots,f_k$ be $k$ models independently 
    trained for different tasks on $\mathcal{D}$. 
    Each model induces a channel that maps the input $x$ to an output.
    Under black-box access, the attacker observes the joint output vector across these channels:
    \begin{equation}\label{eq:joint-obs}
      O(z) \;=\; 
      \bigl(f_1(x),\;\ldots,\;f_k(x)\bigr).
    \end{equation}
    %which constitutes a joint observation over multiple channels associated with the same input.
\end{definition}

\begin{definition}[Membership Inference in ODMM]\label{def:mia-game}
Given the joint output $O(z)$ from Definition~\ref{def:odmm}, the membership inference (MI) is a binary hypothesis test:
\[
    H_1:\; z \in \mathcal{D} \quad\text{vs.}\quad H_0:\; z \notin \mathcal{D}.
\]
Let $\mathcal{L}(\cdot \mid \cdot)$ denote the conditional distribution. Denote $P_{\mathrm{in}}^{(k)} = \mathcal{L}\bigl(O(z) \mid z \in \mathcal{D}\bigr)$ and $P_{\mathrm{out}}^{(k)} = \mathcal{L}\bigl(O(z) \mid z \notin \mathcal{D}\bigr)$ for the output distributions under $H_1$ and $H_0$, respectively. 
%The attacker's goal is to distinguish $P_{\mathrm{in}}^{(k)}$ from $P_{\mathrm{out}}^{(k)}$ given a single realization of $O(z)$.
\end{definition}

The two definitions above establish that the MIA in the ODMM setting can be viewed as an empirical approximation to the Bayes-optimal test over the joint output $O(z)$. Therefore, to formally characterize ODMM exacerbate privacy under composition, it suffices to show that access to more ODMM models strictly increases the attacker's ability to distinguish $P_{\mathrm{in}}^{(k)}$ from $P_{\mathrm{out}}^{(k)}$, surpassing what is achievable from any single model output.

\subsection{Monotonicity of Optimal MI Advantage}
Here, we define the MI Advantage as a measure of distinguishability between $P_{\mathrm{in}}^{(k)}$ and $P_{\mathrm{out}}^{(k)}$, and then establish its monotonicity with respect to the number of observed models in ODMM.

\begin{definition}[MI Advantage]
\label{def:advantage}
For an attacker $\mathcal{A}$ observing the joint output of $k$ models, the membership inference advantage is defined as:
\begin{equation}\label{eq:advantage}
    \scalebox{0.85}{$
    \mathrm{Adv}(\mathcal{A}_k) =
    \bigl|\Pr[\mathcal{A}_k(O_{1:k}) = 1 \mid z \in \mathcal{D}]
    - \Pr[\mathcal{A}_k(O_{1:k}) = 1 \mid z \notin \mathcal{D}]\bigr|
    $}
\end{equation}
where $O_{1:k} = (f_1(x), \ldots, f_k(x))$. The 
Bayes-optimal membership inference advantage is 
$\mathrm{Adv}(\mathcal{A}_k^*) = 
\sup_{\mathcal{A}_k} \mathrm{Adv}(\mathcal{A}_k)$.
\end{definition}

The MI advantage measures the ability of an attacker to distinguish between $P_{\mathrm{in}}^{(k)}$ and $P_{\mathrm{out}}^{(k)}$ based on the observed joint outputs. 
It therefore quantifies the performance of the hypothesis test defined in Definition~\ref{def:mia-game}, with higher values indicating more membership leakage from the joint outputs.
%It therefore quantifies the performance of the hypothesis test defined in Definition~\ref{def:mia-game}, with the Bayes-optimal case corresponding to the optimal statistical distinguishability between the two distributions.

\begin{theorem}[Monotonicity of Optimal MI Advantage]\label{thm:monotone}
For the Bayes-optimal attacker, introducing an additional model in the ODMM setting cannot decrease the optimal membership inference advantage:
\begin{equation}\label{eq:monotone}
    \mathrm{Adv}(\mathcal{A}_{k+1}^*) 
    \;\geq\; \mathrm{Adv}(\mathcal{A}_{k}^*).
\end{equation}
\end{theorem}

\noindent \emph{Proof.} See Appendix~\ref{app:monotone}. 
Theorem~\ref{thm:monotone} holds without additional assumptions. It implies that evaluating ODMM models in isolation, as standard in prior MIA literature, inevitably underestimates the cumulative privacy exposure.

%\subsection{Leakage Accumulation under Conditional Independence}
\subsection{Additive Leakage under Conditional Independence}
\label{sec:additivity}

Theorem~\ref{thm:monotone} only guarantees that the optimal MI advantage is non-decreasing as more ODMM models are exposed. Here, we further prove that, under the assumption of approximate independence among per-model leakage, the privacy leakage from each ODMM model accumulates additively.

\begin{assumption}[Conditional Independence of Leakage]\label{asmp:cond-indep}
Each model $f_i$ produces a per-sample leakage $S_i(z)$ derived from its black-box output. Conditional on the membership status $M \in \{0,1\}$, these leakages are approximately independent:$
  S_1(z), \ldots, S_k(z) 
  \;\perp\!\!\!\perp\; 
  \mid\; M.$
\end{assumption}

Assumption~\ref{asmp:cond-indep} represents an idealized scenario in which each model's leakage is driven solely by its own task-specific memorization, with no inter-task correlation. 
% Under this condition, we then prove that the optimal MI advantage is strictly increasing with the number of exposed models.

\begin{theorem}[Additivity of Log-Likelihood Ratio]\label{thm:llr-decomp}
By the Neyman--Pearson lemma, the optimal attack for distinguishing $P_{\mathrm{in}}^{(k)}$ from $P_{\mathrm{out}}^{(k)}$ is given by a threshold test on the joint log-likelihood ratio $\Lambda_k(z)$. Consequently, the optimal MI advantage $\mathrm{Adv}(\mathcal{A}_k^*)$ is characterized by the distribution of $\Lambda_k(z)$.
Under Assumption~\ref{asmp:cond-indep}, $\Lambda_k(z)$ decomposes additively as:
\begin{equation}
\label{eq:llr}
  \Lambda_k(z) \;=\; \sum_{i=1}^{k} 
  \log \frac{p_i(S_i(z) \mid z \in \mathcal{D})}
  {p_i(S_i(z) \mid z \notin \mathcal{D})}.
\end{equation}
where $p_i(\cdot \mid \cdot)$ denotes the distribution of model's leakage. Under Assumption~\ref{asmp:cond-indep}, each additional model contributes an independent non-negative term to $\Lambda_k$, increasing the separability between $P_{\mathrm{in}}^{(k)}$ and $P_{\mathrm{out}}^{(k)}$ as $k$ grows. As a result, the optimal MI advantage $\mathrm{Adv}(\mathcal{A}_k^*)$ increases monotonically with $k$.
%This additive decomposition implies that $\Lambda_k(z)$ aggregates independent evidence across models, increasing its discriminative power between members and non-members, and consequently $\mathrm{Adv}(\mathcal{A}_k^*)$ grows monotonically with $k$.
\end{theorem}
%\begin{theorem}[Additivity of MI Advantage]\label{thm:llr-decomp}
%By the Neyman--Pearson lemma, the optimal MI advantage $\mathrm{Adv}(\mathcal{A}_k^*)$ is determined by the joint log-likelihood ratio $\Lambda_k(z)$ for distinguishing $P_{\mathrm{in}}^{(k)}$ from $P_{\mathrm{out}}^{(k)}$. Under Assumption~\ref{asmp:cond-indep}, $\Lambda_k(z)$ decomposes additively as:

\noindent \emph{Proof.} See Appendix~\ref{app:llr-decomp}. 
Theorem~\ref{thm:llr-decomp} reveals that ODMM privacy risk amplification arises from aggregating membership leakage information across tasks, similar to the accumulation of privacy loss in the DP composition theorem. Even when each ODMM model provides only a weak MIA signal, the joint log-likelihood ratio can act as a strong discriminator. 

\noindent\textbf{Empirical evidence.} This is supported by the experiments (Section~\ref{sec:ablation_num_models}), where AUC increases monotonically as the number of exposed models grows from 3 to 6 (0.656 $\to$ 0.706 $\to$ 0.743 $\to$ 0.782). %with each additional model contributing a positive incremental gain. 
These results confirm that each additional model monotonically accumulates privacy leakage that composes into an additive membership risk.

\subsection{Amplification Scaling with Task Correlation}
\label{sec:diversity}

Theorem~\ref{thm:llr-decomp} establishes that privacy leakage accumulates additively under the idealized assumption of conditional independence. In practice, however, ODMM models trained on the same dataset for different tasks inevitably exhibit correlated memorization patterns. This is fundamentally different from per-query DP analysis, where each query output is protected by independently and identically sampled noise, typically drawn from Gaussian or Laplace distributions. By recognizing this distinction, we further analyze how inter-task correlations modulate the degree of privacy amplification.

\begin{theorem}[Amplification Degree with Task Correlation]\label{thm:keff}
Let $w_i$ denote each model's standalone leakage strength, and let $\rho_{ij}$ denote the pairwise correlation between models $i$ and $j$'s leakage. The overall privacy amplification degree from $k$ ODMM models is quantified by:
\begin{equation}\label{eq:keff}
    \alpha \;=\; 
    \frac{\bigl(\sum_{i=1}^{k} 
    w_i\bigr)^2}{\sum_{i=1}^{k}\sum_{j=1}^{k} 
    w_i\, w_j\, \rho_{ij}}.
\end{equation}
%where $\alpha$ quantifies the cumulative amplification after accounting for inter-task correlations. 
It follows that $1 \le \alpha \le k$, with larger $\alpha$ indicating higher amplification effects. Specifically, $ \alpha = k$ when all task-specific models are perfectly diverse ($\rho_{ij} = 0,\, i \neq j$), consistent with the scenario in Theorem~\ref{thm:llr-decomp}; and $1 < \alpha < k$ in the realistic setting, where each model contributes positively but with magnitude governed by the divergence of its memorization patterns relative to the existing task ensemble ( $0 < \rho_{ij} < 1$).
\end{theorem}

\noindent \emph{Proof.} See Appendix~\ref{app:keff}. 
Theorem~\ref{thm:keff} bridges the gap between the idealized additivity in Theorem~\ref{thm:llr-decomp} and real-world ODMM deployments considered in our work, where task correlations are inherently non-zero.
It identifies task divergence as the key factor governing privacy risk amplification, implying that combining \textit{heterogeneous (less correlated) tasks} (lower $\rho_{ij}$, higher $\alpha$) leads to more leakage than homogeneous ones.

\noindent\textbf{Empirical evidence.} 
This is supported by the experiments (Section~\ref{sec:divergence}), where task correlation $\rho_{ij}$ is operationalized as the cosine distance between two ODMM models' mean 512-dim penultimate-layer avgpool representations, and the per-pair amplification is measured by $\Delta\mathrm{AUC}(f_{\rm anchor}, f_j) = \mathrm{AUC}(f_{\rm anchor}{+}f_j) - \mathrm{AUC}(f_{\rm anchor})$. On seven CelebA ResNet-18 models with Male (gender classification) as the anchor, the most divergent partners (Nose, SHair) yield the largest amplifications ($\Delta\mathrm{AUC}$ +0.1584, +0.1432), while the least divergent (Hat) yields the smallest (+0.0165). These results confirm that more divergent tasks contribute more non-redundant leakage, substantially increasing the privacy amplification degree $\alpha$.

\subsection{Theoretical Foundations of \name}
Theorems~\ref{thm:monotone}--\ref{thm:keff} suggest that privacy leakage accumulates across ODMM models, and the extent of this amplification is governed by inter-task correlation. Motivated by these insights, we propose \name with two key designs:

\mypara{Task-Specific Feature Extraction (Section~\ref{sec:feature}):}
Since leakage strength depends on task-specific memorization and their correlations, we extract task-tailored features to preserve discriminative signals. This design captures heterogeneous leakage patterns across models, ensuring that non-redundant (i.e., weakly correlated) signals are retained.

\mypara{Joint Information Fusion (Section~\ref{sec:joint}):}
To exploit the compositional nature of leakage, we aggregate task-specific features into a joint representation that approximates the cumulative evidence across models. This fusion enables the attacker to combine complementary leakage signals, effectively amplifying the overall membership inference power.

% ======================================================
\section{Threat Model and Methodology}
\label{sec:methodology}
% ======================================================
% ======================================================

\subsection{Threat Model}
\label{threat}
% ======================================================
% Below defines the considered threat model of \name by specifying the adversary's knowledge, capability, and objective.

\mypara{Adversary Knowledge.}
We assume that the adversary can identify that ODMM models share the same data source through several realistic avenues.
First, when a service provider releases multiple task-specific models in a public model zoo, the training dataset is typically disclosed with documentation. For example, Facebook AI Research's Detectron2~\cite{detectron2} explicitly documents that its object detection, instance segmentation, and keypoint estimation models are all trained on the COCO~\cite{lin2014microsoft}. Second, model cards on platforms such as Hugging Face routinely expose metadata, e.g., the training dataset. For instance, Hugging Face Inference API hosts both question answering~\cite{roberta_squad2_deepset} and question generation~\cite{t5_qa_qg_valhalla} models developed by different providers, which explicitly report SQuAD~\cite{rajpurkar2016squad} as their training dataset.
%For instance, \textit{deepset/roberta-base-squad2}~\cite{roberta_squad2_deepset} for question answering and \textit{valhalla/t5-base-qa-qg-hl}~\cite{t5_qa_qg_valhalla} for question generation, developed by different providers, both explicitly report SQuAD~\cite{rajpurkar2016squad} as their training dataset, and the models are directly accessible via the Hugging Face Inference API.
Following the widely adopted black-box MIAs setting~\cite{pang2025black,chen2024slmia,li2024seqmia,he2024difficulty,meeus2024did}, we assume that the adversary can query ODMM models and obtain only their posterior probability distributions.
Following~\cite{li2024seqmia, he2024difficulty, ZarifzadehLS24}, the adversary is further assumed to possess a local shadow dataset $\mathcal{D}_{s}$ that has the same data distribution as the victim’s dataset but without overlapped samples. 
%\garrison{Here, we need to identify how the adversary can know those models are trained on the same dataset? One scenario is when APIs from multiple ODMM models are released by the same provider. It is more likely that the dataset is reused. In other cases, the model provider may provide meta information about the model, such as when it is updated, which dataset is used. Can we figure out if this is the case through real-world examples?}
% Moreover, the adversary is aware of the architectures used for each task-specific model, which can be stolen in a black-box manner via side-channel information~\cite{yan2020cache,gao2024deeptheft}.

\mypara{Adversary Capability.} The adversary can neither access the internal parameters of ODMM models nor tamper with the underlying model before or after its deployment, such as data poisoning~\cite{ma2024watch} or fault injection~\cite{li2024yes}. However, the adversary can leverage $\mathcal{D}_{s}$ to train a set of shadow models corresponding to each task-specific victim ODMM model.
 
\mypara{Adversary Goal.}
For a given target sample, the adversary aims to infer whether it is a member of a victim dataset used to train ODMM models.

%\subsection{Insight}
% ======================================================
%The core premise of our work relies on the hypothesis of \textbf{memorization complementarity} among ODMM models trained for different tasks on the same data source. As exempfilied in Figure~\ref{fig:insight}, while existing MIAs typically target a single model $f_{\theta}$ performing a specific task $\mathcal{T}$, applications can often involve a dataset $\mathcal{D}$ being repurposed to train multiple models $\{f_{\theta_1}, f_{\theta_2}, ..., f_{\theta_k}\}$ for distinct tasks $\{\mathcal{T}_1, \mathcal{T}_2, ..., \mathcal{T}_k\}$ (e.g., age regression, gender classification, and race detection on face images).

\begin{figure}[t]
    \centering
    \includegraphics[width=\linewidth]{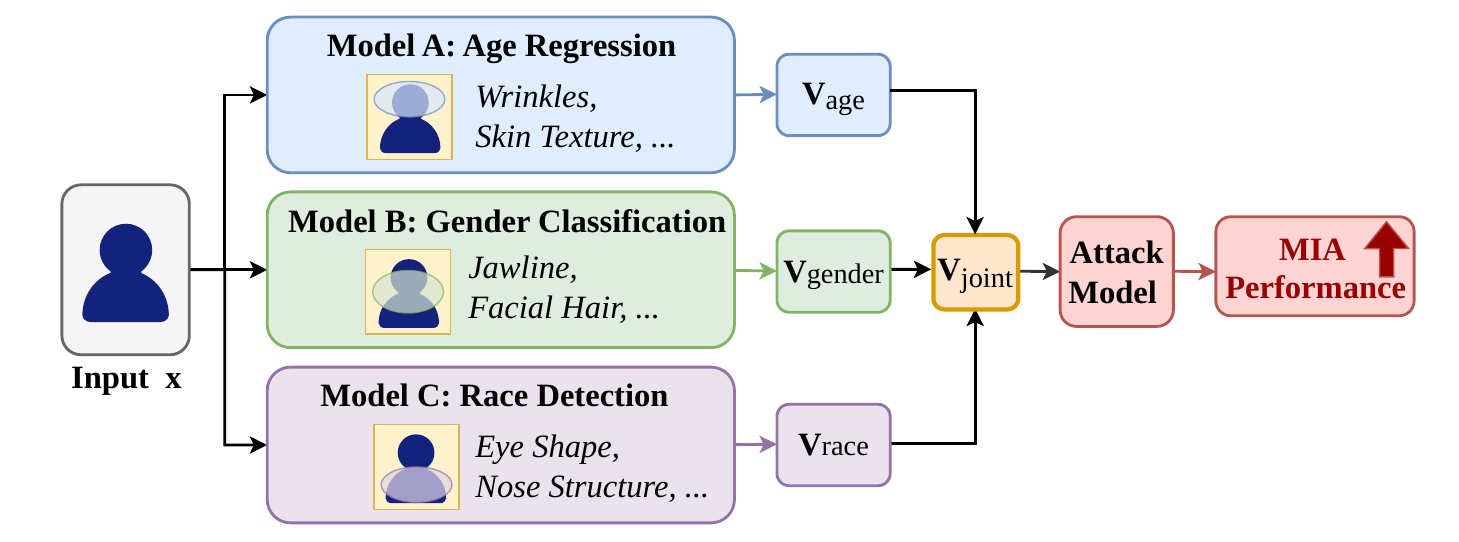}
    \caption{Overview of \name that builds on ODMM privacy composition.%\garrison{increase font size}
    % Given the same input sample $x$, different task-specific ODMM models focus on distinct feature subspaces (e.g., wrinkles for age, jawline for gender, eye shape for race). Each model produces a feature vector $\mathbf{v}_k$ capturing task-specific membership signals. By fusing these vectors into $\mathbf{V}_{\rm joint}$ to exploit the privacy composition, MIA accuracy amplifies.
    }
    \label{fig:insight}
\end{figure}

%We posit that \textit{different tasks force corresponding ODMM models to focus on varying feature subspaces of the data, rendering non-redundant but complementary privacy leakage}.

%\noindent$\bullet$\textbf{Task-Dependent Memorization}. A model's decision boundary is shaped by its optimization objective. For a complex training sample $x$, Task A (e.g., gender classification) might find $x$ easy to generalize, exhibiting low loss and ``forgetting'' the specific details of $x$. However, Task B (e.g., age regression) might find the same sample $x$ effectively ``hard,'' requiring the model to memorize its idiosyncratic features to minimize the training error. Consequently, $x$ might be exposed as a member by Model B but not by Model A.

%\noindent$\bullet$\textbf{Risk Amplification}. By aggregating the inference signals from multiple tasks, an attacker gains a panoramic view of the data point's status. The collective information vector captures membership traces that are otherwise absent in isolation when only a subset of models (e.g., a single model in an extreme case that is commonly studied in existing literature) is present. Therefore, we hypothesize that the joint privacy risk $\mathcal{R}_{\rm joint}$ strictly exacerbate the risk of any \textit{single} model (e.g., $\mathcal{R}_{\mathcal{T}_1}$ or $\mathcal{R}_{\mathcal{T}_2}$), expressed as: 
%\begin{equation}
%\mathcal{R}_{\rm joint} > \textsf{max}(\mathcal{R}_{\mathcal{T}_1}, \mathcal{R}_{\mathcal{T}_2}, ...).
%\end{equation}

% ======================================================
\subsection{The \name Framework}\label{sec:framework}
% ======================================================
Based on the \textit{ODMM privacy composition} theory in Section~\ref{sec:theory}, we propose \name, a unified framework for systematically conducting joint MIA against multiple APIs, each corresponding to an ODMM model, as shown in \ref{fig:insight}. \name has three phases: \Circled{1} Task-Specific Feature Extraction, \Circled{2} Joint Information Fusion, and \Circled{3} Shadow-Assisted Inference.

\subsubsection{Task-Specific Feature Extraction}
\label{sec:feature}
Consider $K$ models trained on the same private dataset $\mathcal{D}_{\rm private}$ can be queried by the adversary. For a target sample $(x, y)$, the attacker queries each model $f_k$ and obtains the raw outputs. To capture the subtle differences in how members and non-members behave across different OODM tasks (i.e., age regression, binary gender classification, multi-class race classification), we design specific feature extractors $\psi_k(\cdot)$.

Based on model output observations, we extract the following multidimensional feature vectors $\mathbf{v}_k$ for each task type:

\noindent\textbf{Regression Tasks} (e.g., Age Estimation)
For a target sample with ground-truth label $y$, a regression model $f_{\rm reg}$ returns a continuous predicted value $\hat{y}$. We construct a 3-dimensional feature vector focusing on prediction error:
$$\mathbf{v}_{\rm reg} = \left[ \hat{y}, \quad \mathcal{L}_{\rm MSE}(y, \hat{y}), \quad |y - \hat{y}| \right],$$
where $\mathcal{L}_{\rm MSE}$ is the mean squared error (MSE) loss. The absolute error term $|y - \hat{y}|$ is \textit{explicitly included} by our \name as it directly correlates with the model's overfitting on outliers.

\noindent\textbf{Binary Classification Tasks} (e.g., Gender Recognition)
A binary classifier $f_{\rm bin}$ outputs a logit $z$. We compute the probability $p = \sigma(z)$ (Sigmoid) and extract a concatenated vector:
$$\mathbf{v}_{\rm bin} = \left[ p, \quad \mathcal{L}_{\rm BCE}(y, p), \quad |p - 0.5| \right].$$
Here, $|p - 0.5|$ represents the prediction confidence. Member samples often exhibit higher confidence.

\noindent\textbf{Multi-class Classification Tasks} (e.g., Race Classification)
For a classifier $f_{\rm mul}$ with $C$ classes, outputting a probability vector $\mathbf{p} = \textsf{Softmax}(\mathbf{z})$, the information is richer. We construct a $(C+3)$-dimensional vector:
$$\mathbf{v}_{\rm mul} = \left[ \mathbf{p}, \quad \mathcal{L}_{\rm CE}(y, \mathbf{p}), \quad \max(\mathbf{p}), \quad \mathcal{H}(\mathbf{p}) \right],$$
where $\max(\mathbf{p})$ captures the confidence of the predicted class, and $\mathcal{H}(\mathbf{p}) = -\sum p_i \log p_i$ is the prediction entropy. Lower entropy typically signals that the model has ``seen'' the sample during training and is confident in its prediction.

\subsubsection{Joint Information Fusion}
\label{sec:joint}
After querying all $K$ available ODMM models, the attacker possesses a set of feature vectors $\{\mathbf{v}_1, \mathbf{v}_2, ..., \mathbf{v}_K\}$ explained aforesaid. The fusion module aggregates these signals to form a single representation. Our \name framework employs a Concatenation Fusion strategy, which preserves the raw inference signals from all tasks without loss of information. The joint feature vector $\mathbf{V}_{\rm joint}$ is defined as:

$$\mathbf{V}_{\rm joint} = \mathbf{v}_1 || \mathbf{v}_2 || ... || \mathbf{v}_K.$$

Take a later studied case to ease understanding. In our UTKFace setup involving Age (Regression), Gender (Binary), and Race (5 class), the resulting joint vector has a dimensionality of $3 + 3 + (5+3) = 14$. This vector creates a high-dimensional membership footprint of the sample $x$, where different dimensions absorb membership signals from different ``perspectives'' (e.g., entropy from race classification vs. MSE from age regression).

\subsubsection{Shadow Model Assisted Attack Inference}
To determine the membership status of $x$ based on $\mathbf{V}_{\rm joint}$, \name employs a shadow model technique. Its ODMM scenario constructions are as follows.

\noindent\textbf{Task-Specific Shadow Training}. For each target task $\mathcal{T}_k$, where $k \in \{1, ..., K\}$, we independently train a set of $N$ shadow models $\{S_{k,1}, S_{k,2}, ..., S_{k,N}\}$ per task. In our experiments, we set $N=5$. Crucially, to simulate the joint exposure scenario, we enforce \textit{data split consistency} across tasks. The $i$-th shadow model of every task (e.g., $S_{\rm Age, i}$, $S_{\rm Gender, i}$, $S_{\rm Race, i}$) is trained on the exact same shadow training subset $\mathcal{D}_{\rm shadow\_train}^{(i)}$ and evaluated on the same non-member subset $\mathcal{D}_{\rm shadow\_test}^{(i)}$. This alignment ensures that a sample $x$ has a consistent ground-truth membership label $m \in \{0, 1\}$ across all tasks within the same shadow index $i$.

\noindent\textbf{Joint Attack Dataset Construction}. We construct the training dataset for the attack model by aggregating outputs from these aligned shadow models. For each shadow index $i \in \{1, ..., N\}$, we feed the corresponding samples into the shadow model group $\{S_{1,i}, S_{2,i}, ..., S_{K,i}\}$. We extract the task-specific feature vectors $\{\mathbf{v}_{1,i}, \mathbf{v}_{2,i}, ..., \mathbf{v}_{K,i}\}$ and fuse them to form the joint training samples:
    $$\mathbf{V}_{\rm joint}^{(i)} = \mathbf{v}_{1,i} \space || \space \mathbf{v}_{2,i} \space || \space ... \space || \space \mathbf{v}_{K,i}.$$

    The final attack dataset is the union of features collected from all $N$ shadow groups:

    $$\mathcal{D}_{\rm attack} = \bigcup_{i=1}^{N} \left\{ (\mathbf{V}_{\rm joint}^{(i)}(x), m_x) \mid x \in \mathcal{D}_{\rm shadow}^{(i)} \right\},$$
where $m_x=1$ if $x$ is a member of the $i$-th split, and $0$ otherwise.

\noindent\textbf{Attack Model Inference}. A binary attack meta-classifier $\mathcal{A}$ (implemented as a Random Forest) is trained on $\mathcal{D}_{\rm attack}$ to distinguish members from non-members based on the fused ODMM features. The Random Forest classifier is selected for its robustness in handling high-dimensional, heterogeneous feature spaces formed by the concatenation of regression losses, classification entropies, and confidence scores. During the attack phase, the attacker queries the victim's API endpoints to obtain the joint feature vector $\mathbf{V}_{\rm victim}$ and computes the final membership score:

$$\textsf{Score}(x) = \mathcal{A}(\mathbf{V}_{\rm victim}) \in [0, 1],$$
where the score represents the predicted probability 
of $x$ being a member. A higher score indicates stronger 
evidence of membership. In practice, the continuous score 
is used for threshold-free evaluation via ROC analysis and 
AUC, which provide a more comprehensive assessment of 
attack effectiveness across all operating points.

% ======================================================
\section{Evaluation}
\label{sec:evaluation}
% ======================================================

We conduct comprehensive experiments to systematically evaluate the effectiveness of \name across diverse domains and model architectures. We first present our general experimental setup (Section~\ref{sec:eval_setup}), then validate our approach on the facial scenario with an in-depth analysis that empirically confirms the ODMM privacy composition theory (Section~\ref{sec:eval_face}). We further extend our evaluation to medical imaging scenarios with heterogeneous architectures (Section~\ref{sec:eval_medical}) and LLM instruction fine-tuning scenarios (Section~\ref{sec:eval_llm}).
We also present a comparison between shared-backbone multi-task learning and independent single-task models (i.e, ODMM) (Section~\ref{sec:discussion_mtl}). Finally, we examine how the number of exposed task-specific models influences \name effectiveness (Section~\ref{sec:ablation_num_models}), corresponding to Theorem~\ref{thm:llr-decomp}.
%between multi-task learning with shared backbones and independent single-task models
% ======================================================
\subsection{Experimental Setup}
\label{sec:eval_setup}
% ======================================================

\mypara{Attack Pipeline.}
For each privacy-sensitive evaluation scenario (i.e., Facial, Medical, and LLM Fine-Tuning), we train $N=5$ shadow models per task. 
% with consistent data splits across all tasks, ensuring that the same samples have identical membership status across different task models.
Task-specific feature vectors are extracted based on output types: 3-dimensional vectors for regression and binary classification tasks (prediction, loss, confidence/error), and $(C+3)$-dimensional vectors for $C$-class classification tasks (class probabilities, loss, max probability, entropy). All other general settings follow Section~\ref{sec:framework}.
% Joint features are formed by concatenation across tasks. We employ Random Forest classifiers as attack meta-classifiers throughout all experiments unless otherwise stated.

\mypara{Evaluation Metrics.}
We use standard MIA evaluation metrics: Attack Accuracy (Acc), Precision (Prec), Recall, F1 Score, and Area Under the ROC Curve (AUC). We also report True Positive Rate at low False Positive Rates (TPR@FPR) where relevant, as it is critical for practical privacy auditing.

% ======================================================
\subsection{Facial Scenario Analysis}
\label{sec:eval_face}
% ======================================================

Facial analysis represents a canonical ODMM scenario, where the same face images are routinely repurposed for multiple prediction tasks such as age estimation, gender classification, and ethnicity recognition. We here evaluate \name in-depth on three representative facial attribute datasets spanning different scales, task compositions, and model architectures.

\subsubsection{Datasets, Tasks, and Model Configuration}
\label{sec:eval_face_setup}

\mypara{Datasets and Tasks.}
UTKFace~\cite{zhifei2017cvpr} comprises over 23,000 face images annotated with age, gender, and ethnicity attributes, representing a realistic ODMM scenario with heterogeneous task types. We configure three tasks: \textit{Age} (regression), \textit{Gender} (binary classification), and \textit{Race} (5-class classification).
\textbf{CelebA}~\cite{liu2015faceattributes} has over 200,000 celebrity face images with 40 binary attribute annotations. We select three binary classification tasks: \textit{Male} (gender), \textit{Mouth\_Slightly\_Open} (expression), and \textit{Big\_Nose} (facial feature).
\textbf{FairFace}~\cite{karkkainenfairface} has 108,000 images with balanced demographic distributions. We employ three multi-class classification tasks: \textit{Race} (7-class), \textit{Age} (9-class), and \textit{Gender} (binary).

\mypara{Model Architectures.}
For UTKFace and CelebA, we employ ResNet-18~\cite{he2016deep} as the backbone architecture. For FairFace, we adopt Vision Transformer (ViT-Base)~\cite{dosovitskiy2021an} with LoRA~\cite{hu2021loralowrankadaptationlarge} fine-tuning ($r=8$, $\alpha=16$), representing a typical parameter-efficient fine-tuning (PEFT) paradigm. Note that within each dataset, all ODMM models share the same backbone architecture (i.e., ResNet-18 for UTKFace/CelebA and ViT for FairFace). The setting with heterogeneous backbones across ODMM models is explored in Section~\ref{sec:eval_medical}. All models are trained for 50 epochs using Adam/AdamW optimizer with learning rate $10^{-3}$ (or $10^{-4}$ for ViT).

\mypara{Joint Features.}
The concatenated joint feature vectors yield 9d (CelebA: $3 \times 3$), 14d (UTKFace: $3 + 3 + 8$), and 25d (FairFace: $2 + 12 + 10 + 1$) representations, respectively.

\subsubsection{Theorem~\ref{thm:keff} Evidence}
\label{sec:divergence}

Theorem~\ref{thm:keff} predicts that the privacy amplification degree $\alpha$ is governed by inter-task correlation: \emph{the more divergent (less correlated) two tasks are, the greater the joint leakage they yield}. To empirically validate this prediction, we conduct a controlled pairwise study on CelebA with 7 representative binary classification tasks sharing an identical ResNet-18 backbone and training configuration: \textit{Male}, \textit{Mouth\_Slightly\_Open} (Mouth), \textit{Big\_Nose} (Nose), \textit{Black\_Hair} (BHair), \textit{Smiling} (Smile), \textit{Straight\_Hair} (SHair), and \textit{Wearing\_Hat} (Hat). Fixing the architecture and dataset isolates inter-task divergence as the only varying factor.

\noindent\textbf{Quantifying Task Divergence.}
Since $\rho_{ij}$ in Theorem~\ref{thm:keff} reflects how similarly two models memorize the shared samples, we operationalize it through their internal representations. For each victim model $f_i$, we extract the 512-dim feature vector from its penultimate \texttt{avgpool} layer for every training (member) sample, and compute the mean representation $\bar{\mathbf{h}}_i$. The pairwise task divergence between $f_i$ and $f_j$ is quantified by the cosine distance:
\begin{equation}
\mathrm{Dist}(f_i, f_j) = 1 - \frac{\bar{\mathbf{h}}_i \cdot \bar{\mathbf{h}}_j}{\|\bar{\mathbf{h}}_i\| \cdot \|\bar{\mathbf{h}}_j\|}.
\end{equation}
A larger $\mathrm{Dist}(f_i, f_j)$ indicates that $f_i$ and $f_j$ encode the same data through more divergent representations, corresponding to a smaller $\rho_{ij}$ in Theorem~\ref{thm:keff}.

\noindent\textbf{Quantifying Privacy Amplification.}
For each pair $(f_{\rm anchor}, f_j)$, we measure the marginal MIA gain that $f_j$ contributes when fused with the anchor $f_{\rm anchor}$:
\begin{equation}
\Delta\mathrm{AUC}(f_{\rm anchor}, f_j) = \mathrm{AUC}(f_{\rm anchor} \!+\! f_j) - \mathrm{AUC}(f_{\rm anchor}),
\end{equation}
where $\mathrm{AUC}(f_{\rm anchor} \!+\! f_j)$ is the \name attack AUC obtained by fusing the two models (Section~\ref{sec:framework}), and $\mathrm{AUC}(f_{\rm anchor})$ is the single-task MIA AUC of the anchor.

\noindent\textbf{Validation Hypothesis.}
If Theorem~\ref{thm:keff} holds, then ideally for any anchor $f_{\rm anchor}$ and any two candidate partners $f_j$, $f_k$:
\begin{equation}
\begin{aligned}
\mathrm{Dist}(f_{\rm anchor}, f_j) > \mathrm{Dist}(f_{\rm anchor}, f_k) \;\to\; \\
\Delta\mathrm{AUC}(f_{\rm anchor}, f_j) > \Delta\mathrm{AUC}(f_{\rm anchor}, f_k).
\end{aligned}
\label{eq:hypothesis}
\end{equation}
That is, fusing the anchor with a more divergent partner expects a larger amplification.

\noindent\textbf{Experimental Result.}
Table~\ref{tab:divergence_hat} reports the result with \textit{Male} as the anchor model. The six candidate partners are listed in descending order of their cosine distance to Male, alongside the resulting $\Delta\mathrm{AUC}$ from joint MIA. The two rankings are nearly perfectly aligned: the most divergent partners (Nose and SHair) jointly contribute the largest amplifications ($+0.1584$ and $+0.1432$), whereas the least divergent partner (Hat) yields the smallest gain ($+0.0165$). Out of $\binom{6}{2}=15$ ordering pairs constructed from these 6 partners, \textbf{14 are consistent} with Eq.~\eqref{eq:hypothesis}; only one exception between Mouth and Smile marginally deviates, where the two distances are closely matched ($0.0670$ vs.\ $0.0650$). This is expected because there are variance factors, such as model accuracy and the partner's single-task MIA vulnerability degree. Specifically, Smile (MIA Acc.=0.5326) itself is more vulnerable to MIA than Mouth (0.0.5264), as detailed in Table~\ref{tab:single_task_celeba} in the Appendix.
% This is potentially caused by training stochastic or a different convergence speed.

\begin{table}[t]
\centering
\caption{Empirical evidence of Theorem~\ref{thm:keff} on CelebA, with \textit{Male} as the anchor model. Each row reports the cosine distance between Male and a partner, and the resulting $\Delta\mathrm{AUC}$ from joint MIA. Partners are sorted by distance.}
\label{tab:divergence_hat}
\resizebox{0.95\linewidth}{!}{
\begin{tabular}{lcc}
\toprule
\textbf{Partner} & \textbf{Dist(Male, Partner)} & $\boldsymbol{\Delta\mathrm{AUC}}$ \\
\midrule
Nose  & 0.0926 & +0.1584 \\
SHair & 0.0799 & +0.1432 \\
BHair & 0.0720 & +0.0993 \\
Mouth & 0.0670 & +0.0526 \\
Smile & 0.0650 & +0.0613 \\
Hat   & 0.0606 & +0.0165 \\
\midrule
\multicolumn{3}{l}{\textit{Distance rank:}\quad Nose $>$ SHair $>$ BHair $>$ Mouth $>$ Smile $>$ Hat} \\
\multicolumn{3}{l}{\textit{$\Delta\mathrm{AUC}$ rank:}\quad Nose $>$ SHair $>$ BHair $>$ Smile $>$ Mouth $>$ Hat} \\
\multicolumn{3}{l}{\textit{Pairwise ordering match: }\textbf{14/15}} \\
\bottomrule
\end{tabular}
}
\end{table}

These observations provide direct empirical support for Theorem~\ref{thm:keff}: more divergent task-specific models, by encoding the shared data through less correlated representations, contribute more non-redundant membership leakage when jointly exposed. Conversely, semantically similar tasks—\textit{e.g.}, \textit{Hat} that captures coarse head-region cues highly correlated with gender-related appearance, or \textit{Mouth} and \textit{Smile} that share heavily overlapping mouth-region features—yield smaller marginal amplification, consistent with the larger $\rho_{ij}$ predicted by the theorem.

\subsubsection{MIA Attack Performance}
\label{sec:eval_face_results}

% We now evaluate the end-to-end attack effectiveness of \name. 
Table~\ref{tab:vision_results} details MIA results of \name and compares them with corresponding single-task attacks.
% and the \name joint attack across all three facial attribute datasets.

\begin{table}[t]
\centering
\caption{MIA on facial scenario. Utility denotes each model's utility, e.g., accuracy or MAE.
% on its original prediction task (validation set accuracy for classification, MAE for regression).
}
\label{tab:vision_results}
\resizebox{\linewidth}{!}{
\begin{tabular}{llcccccc}
\toprule
\textbf{Dataset} & \textbf{Attack Type} & \textbf{Utility} 
& \textbf{MIA Acc.} & \textbf{MIA Prec.} & \textbf{MIA Recall} 
& \textbf{MIA F1} & \textbf{MIA AUC} \\
\midrule
\multirow{5}{*}{\textbf{CelebA}} 
& Single-Male & 98.1\% & 0.512 & 0.512 & 0.519 & 0.516 & 0.521 \\
& Single-Mouth & 98.4\% & 0.530 & 0.525 & 0.627 & 0.572 & 0.544 \\
& Single-Big\_Nose & 98.2\% & 0.591 & 0.587 & 0.608 & 0.598 & 0.649 \\
& \textbf{\name Attack} & -- & \textbf{0.661} & \textbf{0.642} 
& \textbf{0.728} & \textbf{0.682} & \textbf{0.732} \\
\cmidrule{2-8}\rowcolor{green!10}
& \textit{Impr. over Best Single} & -- & \textit{+11.8\%} 
& \textit{+9.4\%} & \textit{+19.7\%} & \textit{+14.0\%} 
& \textit{+12.8\%} \\
\midrule
\multirow{5}{*}{\textbf{UTKFace}} 
& Single-Age & 6.98$^\dagger$ & 0.712 & 0.678 & 0.807 & 0.737 & 0.767 \\
& Single-Gender & 98.5\% & 0.599 & 0.593 & 0.633 & 0.613 & 0.657 \\
& Single-Race & 98.1\% & 0.734 & 0.685 & 0.865 & 0.765 & 0.793 \\
& \textbf{\name Attack} & -- & \textbf{0.869} & \textbf{0.806} 
& \textbf{0.972} & \textbf{0.881} & \textbf{0.925} \\
\cmidrule{2-8}\rowcolor{green!10}
& \textit{Impr. over Best Single} & -- & \textit{+18.4\%} 
& \textit{+17.7\%} & \textit{+12.4\%} & \textit{+15.2\%} 
& \textit{+16.6\%} \\
\midrule
\multirow{5}{*}{\textbf{FairFace}} 
& Single-Race & 72.3\% & 0.687 & 0.622 & 0.951 & 0.752 & 0.745 \\
& Single-Age & 68.5\% & 0.703 & 0.654 & 0.866 & 0.745 & 0.758 \\
& Single-Gender & 91.2\% & 0.551 & 0.529 & 0.950 & 0.679 & 0.577 \\
& \textbf{\name Attack} & -- & \textbf{0.804} & \textbf{0.758} 
& \textbf{0.893} & \textbf{0.820} & \textbf{0.863} \\
\cmidrule{2-8}\rowcolor{green!10}
& \textit{Impr. over Best Single} & -- & \textit{+14.4\%} 
& \textit{+15.9\%} & \textit{-6.1\%} & \textit{+9.0\%} 
& \textit{+13.9\%} \\
\bottomrule
\multicolumn{8}{l}{\small $^\dagger$ MAE for regression; 
lower is better. Improvement is computed relative to the 
\textit{best} single-task baseline per metric.} \\
\end{tabular} 
}
\end{table}

\mypara{Consistent Across Datasets and Architectures.}
The \name consistently outperforms all single-task baselines across three diverse datasets and two distinct model architectures. On CelebA, the \name achieves AUC of 0.732 (\textbf{+12.8\%} over \textbf{best} single-task). On UTKFace with heterogeneous task types, the amplification is most pronounced with AUC reaching 0.925 (\textbf{+16.6\%} over the best). FairFace with ViT+LoRA fine-tuning has similar amplification patterns (AUC: 0.758 $\rightarrow$ 0.863, \textbf{+13.9\%}), indicating that the \name generalizes beyond CNN architectures to modern transformer-based models. 

\mypara{Task/Feature divergence Enhances MIA.}
UTKFace exhibits the strongest privacy amplification, which we attribute to its \textit{heterogeneous task types} (regression + binary + multi-class classification). The regression task captures continuous prediction errors, while multi-class tasks provide rich probability distributions, and these complementary signals jointly reveal membership information invisible to any single task. In contrast, CelebA's three homogeneous binary classification tasks yield comparatively lower yet still substantial amplification. This is consistent with Theorem~\ref{thm:keff}, where heterogeneous task combinations (e.g., Age+Race) yield a more divergent distance than homogeneous ones.

\mypara{Weak Signals Contribute Meaningfully.}
Even tasks with near-random individual attack performance contribute substantially to joint attacks. On CelebA, the Male classifier shows weak MIA performance (AUC=0.521), yet the \name achieves 28.1\% improvement. On FairFace, the Gender task (AUC=0.577) contributes only 9.4\% feature importance but remains essential for the attack. This further validates ODMM privacy composition: different tasks capture \textit{non-redundant} privacy signals that accumulate when fused.

\mypara{Improved Member-Nonmember Balance.}
A notable finding on FairFace is that single-task attacks exhibit severely imbalanced recognition rates (e.g., Gender: 95\% member recognition vs. 15\% non-member recognition), indicating overfitting to membership prediction. The \name dramatically improves non-member recognition (\textbf{+92.2\%} relative improvement: 37.2\% $\rightarrow$ 71.5\%) while maintaining strong member recognition.
% , resulting in more balanced and reliable membership inference.

\begin{figure}[t]
    \centering
    \includegraphics[width=0.95\linewidth]{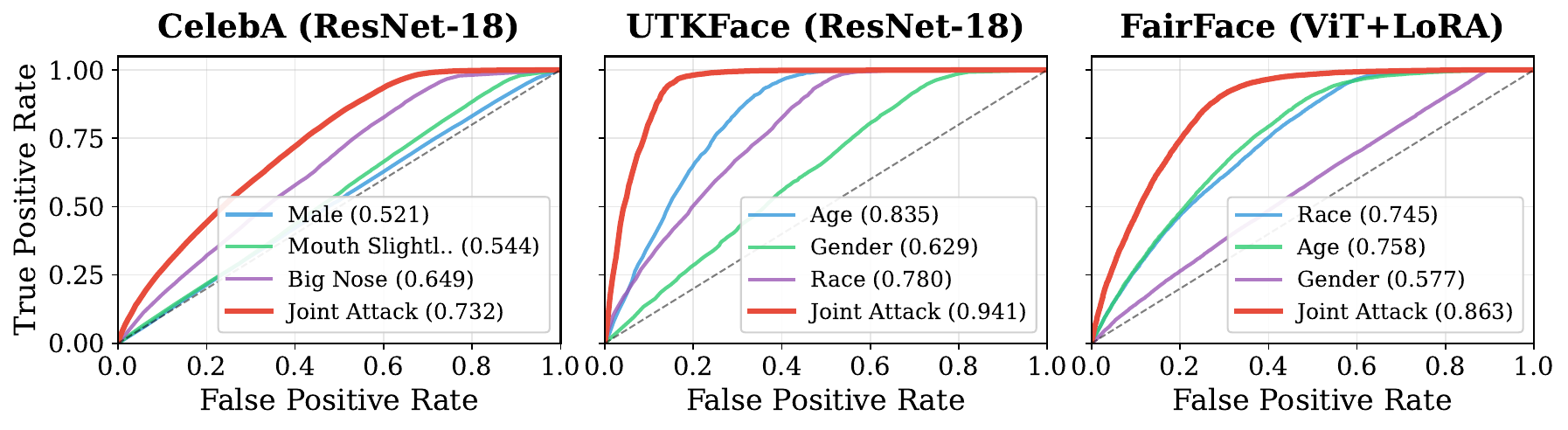}
    \caption{ROC curves comparing single-task attacks versus \name attack across three facial attribute datasets.}
    \label{fig:vision_roc}
\end{figure}

Figure~\ref{fig:vision_roc} further visualizes ROC across all datasets. The \name curves consistently lie far above all single-task curves at any FPR. In the operationally critical low-FPR region (FPR $<$ 0.1), the \name achieves substantially higher TPR, enabling reliable inference with minimal false alarms.

\subsubsection{In-Depth Analysis on UTKFace}
\label{sec:eval_face_indepth}

To gain deeper insight into how multi-task fusion amplifies MIA, we conduct detailed analyses on UTKFace.
% , which serves as a representative ODMM scenario with heterogeneous task types.

\mypara{Confidence Score Distribution.}
Figure~\ref{fig:confidence} shows the distribution of 
attack meta-classifier confidence scores for members and non-members across single-task attacks and the \name.

\begin{figure}[t]
    \centering
    \includegraphics[width=\linewidth]{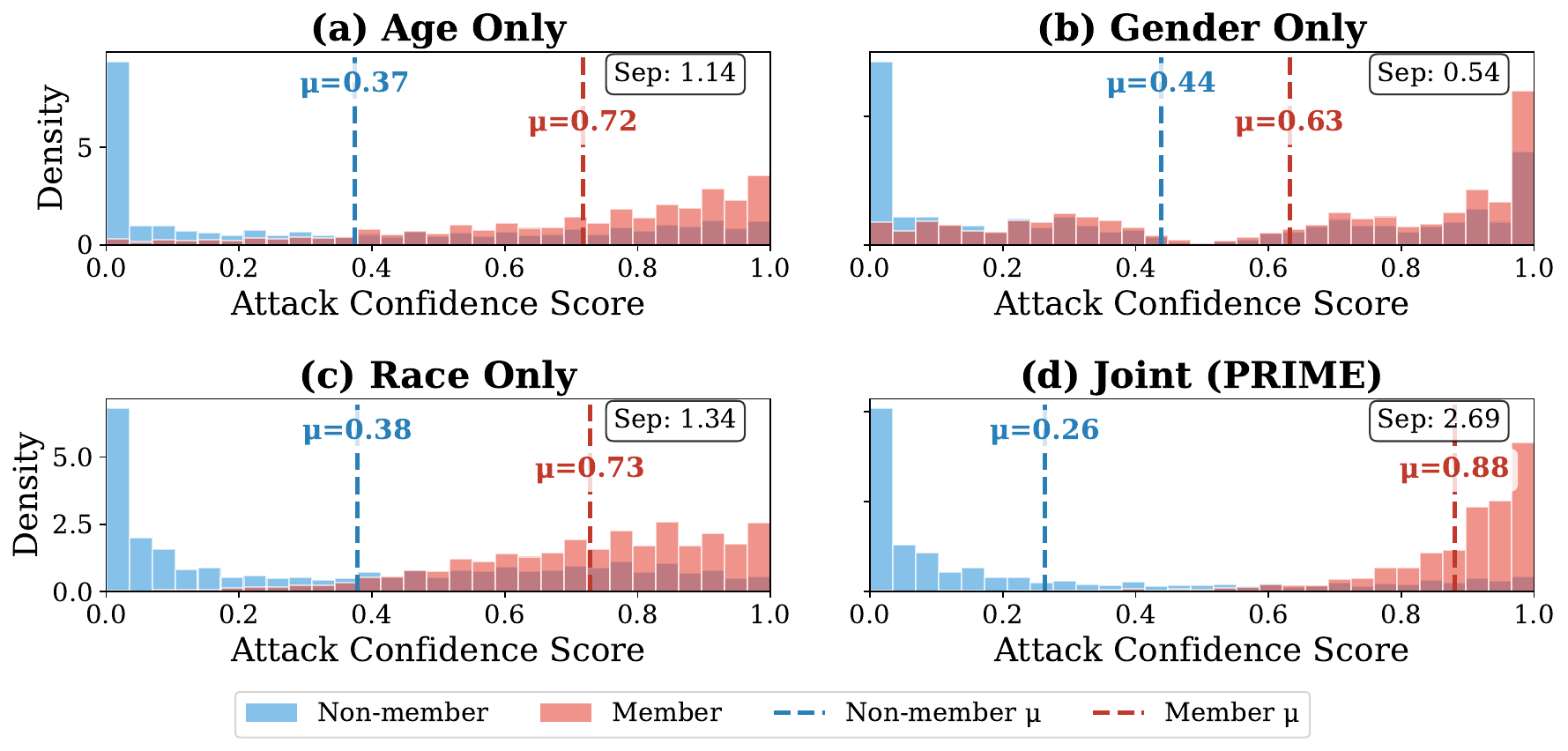}
    \caption{Distribution of attack meta-classifier confidence 
    scores for members (red) and non-members (blue) on UTKFace. 
    Single-task attacks (a--c) exhibit substantial overlap 
    between the two distributions.}
    \label{fig:confidence}
\end{figure}

The three single-task attacks exhibit heterogeneous 
and limited discriminative power:
\textit{Age} achieves moderate separation with member 
mean $\mu_{\rm mem}=0.72$ and non-member mean 
$\mu_{\rm non}=0.37$, but with heavy overlap in the 
0.3--0.6 range.
\textit{Gender} shows the weakest separation 
($\mu_{\rm mem}=0.63$, $\mu_{\rm non}=0.44$), with 
nearly indistinguishable distributions consistent 
with its lowest individual AUC (0.657).
\textit{Race} achieves the best single-task 
discrimination ($\mu_{\rm mem}=0.73$, 
$\mu_{\rm non}=0.38$), yet still exhibits substantial 
distribution overlap.

The \name 
(Figure~\ref{fig:confidence}d) produces a qualitative 
shift in discriminability.
The non-member distribution concentrates below 0.3 
($\mu_{\rm non}=0.26$), while the member distribution 
shifts above 0.8 ($\mu_{\rm mem}=0.88$), yielding 
near-disjoint distributions with minimal overlap.
This stands in stark contrast to the single-task cases.
The resulting AUC improvement from 0.793 (best 
single-task) to 0.925 (\name), as quantified by the ROC analysis in the following, confirms 
that this visual separation translates directly 
into stronger attack performance.

\begin{figure}[t]
    \centering
    \includegraphics[width=0.95\linewidth]{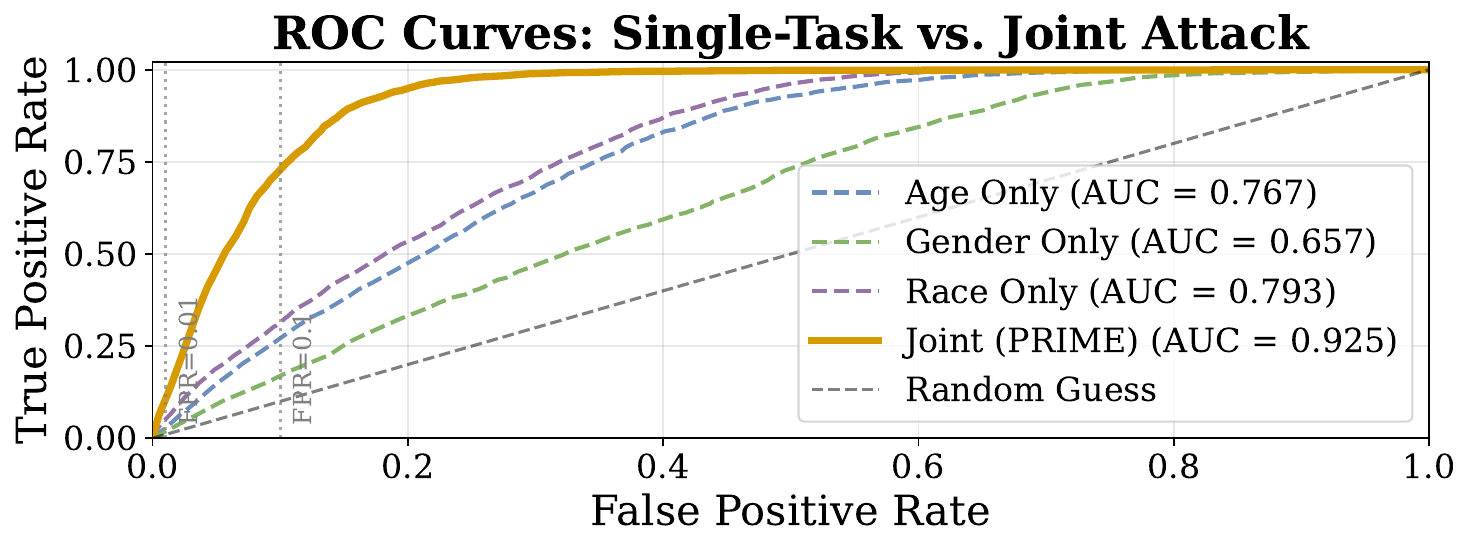}
    \caption{ROC curves (UTKFace).
    \name (AUC=0.925) consistently outperforms all single-task baselines, achieving 16.6\% relative improvement over the \textit{best} single-task attack (Race, AUC=0.793).}
    \label{fig:roc}
\end{figure}

\mypara{ROC Curve Analysis.} Figure~\ref{fig:roc} compares detailed ROC curves on UTKFace.

\noindent$\bullet$\textit{Heterogeneous task-specific vulnerability}: AUC values range from 0.657 (Gender) to 0.793 (Race), with Age at 0.767. This 20\% spread reflects varying memorization patterns induced by different objectives.
    
\noindent$\bullet$\textit{Consistent \name superiority}: The \name curve lies far above all single-task curves at every FPR level, confirming that multi-task fusion provides uniformly stronger attacks rather than task-specific improvements. \name AUC=0.925 represents a \textbf{16.6\% relative improvement} (0.132 absolute) over the \textit{best} single-task baseline.
    
\noindent$\bullet$\textit{High TPR@low-FPR gains}: \name demonstrates substantial advantages in the low FPR regime critical for practical attacks. At FPR=0.01, \name achieves TPR=0.105 compared to TPR=0.051 for the \textit{best} single-task attack (Race), representing a \textbf{2.07$\times$} improvement (+107\% relative gain). At FPR=0.1, the advantage remains pronounced: \name reaches TPR=0.728 versus only TPR=0.315 for best single-task attack, a \textbf{2.32$\times$} improvement.

\subsubsection{Task Contribution Analysis}
\label{sec:eval_face_importance}

Table~\ref{tab:feature_importance} shows the contribution of each task to joint attack effectiveness across all three datasets. Notably, all tasks contribute non-trivially, with the most informative task varying across datasets. This distribution confirms that privacy signals are distributed across tasks rather than concentrated in a single model, consistent with the fact that each ODMM model contributes additivity-like leakage (Theorem~\ref{thm:llr-decomp}).

\begin{table}[t]
\centering
\caption{Feature importance (\%) by task in \name attacks. 
% All tasks contribute non-redundant privacy signals to the attack.
}
\label{tab:feature_importance}
\resizebox{0.90\linewidth}{!}{
\begin{tabular}{lll}
\toprule
\textbf{Dataset} & \textbf{Task} & \textbf{Importance (\%)} \\
\midrule
\multirow{3}{*}{CelebA/UTKFace/ FairFace} 
& Big\_Nose/Race/Age& 47.6/38.2/49.2\\
& Mouth\_Open/Age/Race& 27.6/35.4/41.3\\
& Male/Gender/Gender&24.9/26.4/9.4 \\
% \midrule
% \multirow{3}{*}{UTKFace} 
% & Race & 38.2 \\
% & Age & 35.4 \\
% & Gender & 26.4 \\
% \midrule
% \multirow{3}{*}{FairFace} 
% & Age & 49.2 \\
% & Race & 41.3 \\
% & Gender & 9.4 \\
\bottomrule
\end{tabular}
}
\end{table}

% \subsubsection{Summary of Facial Scenario Analyses}

% These experiments provide compelling evidence that ODMM model exposure substantially amplifies privacy risks. Regardless of dataset scale (23K to 200K images), model architecture (ResNet vs. ViT), training paradigm (full fine-tuning vs. LoRA), or task composition (homogeneous vs. heterogeneous), \name consistently achieves 12.8\% to 16.6\% higher AUC compared to the strongest individual model baseline.

% \begin{mdframed}[backgroundcolor=black!10,rightline=false,leftline=false,topline=false,bottomline=false,roundcorner=2mm]
% \textbf{Key Finding:} ODMM models trained on shared data create correlated but non-redundant privacy leakage. Evaluating each model in isolation---common in existing literature---underestimates the cumulative privacy exposure. ODMM model deployments require holistic privacy auditing that accounts for the amplified risk from collective model access when the same data source is efficiently reused for different task objectives.
% \end{mdframed}

\begin{mdframed}[backgroundcolor=black!10,rightline=false,leftline=false,topline=false,bottomline=false,roundcorner=2mm]
\textbf{Takeaway 1:} Compelling evidence confirms that ODMM model exposure substantially amplifies privacy risks. Regardless of dataset, its scale (23K to 200K images), model architecture (ResNet vs. ViT), training paradigm (full fine-tuning vs. LoRA), or task type composition (homogeneous vs. heterogeneous), \name consistently achieves 12.8\% to 16.6\% higher AUC compared to the \textbf{strongest} individual model baseline.
\end{mdframed}

% ======================================================
\subsection{Medical Imaging Scenario with Heterogeneous Model Architecture}
\label{sec:eval_medical}
% ======================================================

Medical imaging is particularly privacy-sensitive, as patient data often carries highly confidential health information. 
% The exposure of training membership in medical datasets could reveal sensitive diagnoses, treatment histories, or demographic information. 
% We evaluate \name on the HAM10000 dermatology image dataset.

\mypara{Dataset and Tasks.}
HAM10000~\cite{tschandl2018ham10000} is a large-scale dermatoscopic dataset containing 10,015 images of pigmented skin lesions. Each image is annotated with multiple clinical attributes, which can be repurposed for different tasks. We configure three clinically relevant tasks with \textit{heterogeneous} model architectures. \Circled{1} \textit{Diagnosis (dx)} is a 7-class classification (akiec, bcc, bkl, df, mel, nv, vasc) using \textit{ResNet-18}.
\Circled{2} \textit{Age} regression predicts patient age using \textit{VGG-16}. \Circled{3} \textit{Localization} is a 13-class classification of lesion body location using \textit{DenseNet-121}. This setup reflects realistic clinical scenarios where the same dermatoscopic images are analyzed for multiple diagnostic purposes, and different model architectures may be deployed based on task-specific requirements.

\mypara{Lesion-Based Data Splitting.}
A unique characteristic of HAM10000 is that multiple images may correspond to the same skin lesion (identified by \texttt{lesion\_id}). To prevent data leakage between training and evaluation sets, we perform \textit{lesion-level splitting}: all images from the same lesion are assigned to either the member or non-member set, ensuring no overlap at the lesion level. This rigorous splitting strategy makes our privacy evaluation more realistic but challenging.

\mypara{Configuration.}
All models are trained for 50 epochs with Adam optimizer (learning rate $10^{-3}$, batch size 48). We extract task-specific features: 10d for diagnosis (7-class probabilities + loss + max probability + entropy), 3d for age (prediction + loss + error), and 16d for localization (13-class probabilities + loss + max probability + entropy), yielding a 29-dimensional \name joint feature vector.

% \subsubsection{Result and Analysis}
\begin{table}[h]
\centering
\caption{MIA performance (HAM10000). Utility denotes ODMM model utility.
% validation set accuracy for classification or MAE for regression. 
% Improvement is computed relative to the \textit{best} single-task baseline per metric.
}
\label{tab:medical_results}
\resizebox{0.95\linewidth}{!}{
\begin{tabular}{llcccc}
\toprule
\textbf{Attack} & \textbf{Architecture} & \textbf{Utility} 
& \textbf{MIA Acc.} & \textbf{MIA F1} & \textbf{MIA AUC} \\
\midrule
Diagnosis & ResNet-18 & 71.8\% & 0.657 & 0.694 & 0.738 \\
Age & VGG-16 & 13.38$^\dagger$ & 0.516 & 0.520 & 0.521 \\
Localization & DenseNet-121 & 29.1\% & 0.748 & 0.771 & 0.829 \\
\midrule
\textbf{\name} & \textbf{All} & -- & \textbf{0.806} 
& \textbf{0.821} & \textbf{0.886} \\
\midrule
\rowcolor{green!10}
\textit{Impr. over Best Single} & -- & -- 
& \textit{+7.8\%} & \textit{+6.5\%} & \textit{+6.9\%} \\
\bottomrule
\multicolumn{6}{l}{\small $^\dagger$ MAE (Mean Absolute Error) 
for regression task; lower is better.} \\
\end{tabular}
}
\end{table}

\noindent\textbf{Result and Analysis.} Table~\ref{tab:medical_results} presents and compares the MIA performance.
The \name attack achieves AUC of 0.886, representing a \textbf{6.9\%} relative improvement over the \textit{best} single-task attack (0.829).

\noindent$\bullet$\textit{Architecture Heterogeneity.}
Unlike our facial analysis experiments, where tasks within the same dataset shared the same model architecture, HAM10000 employs three distinct CNN architectures (ResNet-18, VGG-16, DenseNet-121). Despite this architectural heterogeneity, the \name attack still achieves substantial privacy amplification. This finding suggests that ODMM privacy composition is primarily dependent on the \textit{data and task} rather than the \textit{model architecture}.

\noindent$\bullet$\textit{Weak Signals Remain Valuable.}
The age regression task exhibits near-random single-task attack performance (AUC = 0.521), suggesting minimal individual privacy leakage. However, feature importance analysis reveals that age-derived features still contribute 6.8\% to the joint attack's predictive power. More importantly, the diagnosis (41.3\%) and localization (52.0\%) tasks provide strong complementary signals that collectively amplify the attack. The localization task, which captures where on the body a lesion appears, proves particularly informative for membership inference, likely because spatial patterns are highly patient-specific.

\noindent$\bullet$\textit{Implications for Clinical AI Deployment.}
These results carry significant implications for healthcare AI systems. Many clinical workflows involve training multiple models on the same patient cohort for different diagnostic or prognostic purposes. Our findings demonstrate that even if individual models appear privacy-safe (i.e., Age with a close to 0.5 AUC) when audited in isolation, their collective deployment creates amplified privacy risks. 
% Healthcare institutions deploying multi-task models on sensitive patient data must consider holistic privacy auditing that accounts for potential joint access to multiple model outputs.

% ======================================================
\subsection{LLM Instruction Fine-Tuning Scenario}
\label{sec:eval_llm}
% ======================================================

The LLMs as foundational models have led to widespread adoption of fine-tuning paradigms, where pre-trained models are adapted to domain-specific tasks using proprietary datasets. We demonstrate that \name extends beyond discriminative vision models, including previously considered ViT, to generative LLMs, revealing that fine-tuning on the same corpus creates substantial privacy amplification risks.

\mypara{Dataset and Tasks.}
We use the Blog Authorship Corpus~\cite{schler2006effects}, containing blog posts annotated with author demographics. We configure three attribute prediction tasks representing diverse prediction types within the generative context. \textit{Gender} is a binary classification (male/female). \textit{Horoscope} is a 12 class classification (zodiac signs). \textit{Age} regression forecasts author's age. Each task fine-tunes the LLM to predict the corresponding attribute based on blog content, simulating realistic scenarios where generative text is analyzed for multiple profiling purposes.

\mypara{Prompt Template and Generation.}
To formulate each prediction task as a generative objective, we design task-specific prompt templates that instruct the LLM to generate the target attribute as free form text. Table~\ref{tab:prompt_template} in the Appendix illustrates the prompt template and example responses for each task. During fine-tuning, the model is trained to generate the response tokens autoregressively, with the prompt tokens masked from the loss computation.

\mypara{Model and Fine-Tuning Configuration.}
We employ Qwen3-1.7B~\cite{qwen3technicalreport} as the base model with LoRA~\cite{hu2021loralowrankadaptationlarge} fine-tuning ($r=16$, $\alpha=32$, dropout=0.05) targeting attention and MLP layers. Each task is fine-tuned independently for 10 epochs with a learning rate of $2 \times 10^{-4}$. We use 10,000 total samples, split equally between target model data and shadow model pool, with 5 shadow models per task.

\mypara{Feature Extraction.}
For LLM-based MIA, we extract 4-dimensional features per task: \textit{loss}, \textit{perplexity}, \textit{confidence} (average maximum token probability), and \textit{entropy} (average token-level entropy). \name concatenates features from multiple tasks, yielding 8d (2-task) or 12d (3-task) feature vectors.

\mypara{Attack Meta-Classifier.}
We evaluate three meta-classifier classifiers: Logistic Regression, Random Forest (RF), and Multi-Layer Perceptron (MLP). Results are reported using the best-performing classifier for each configuration.

% \subsubsection{Result and Analysis}

\begin{table}[t]
\centering
\caption{MIA performance (Qwen3-1.7B Fine-tuning).}
\label{tab:llm_results}
\resizebox{0.75\linewidth}{!}{
\begin{tabular}{lccc}
\toprule
\textbf{Attack Configuration} & \textbf{Accuracy} & \textbf{F1} & \textbf{AUC} \\
\midrule
\multicolumn{4}{l}{\textit{Single-Task Attacks}} \\
\quad Gender & 0.618 & 0.602 & 0.663 \\
\quad Horoscope & 0.787 & 0.779 & 0.832 \\
\quad Age & 0.664 & 0.651 & 0.730 \\
\midrule
\multicolumn{4}{l}{\textit{\name Attacks}} \\
\quad Gender + Horoscope (2-way) & 0.835 & 0.831 & 0.859 \\
\quad Gender + Age (2-way) & 0.764 & 0.758 & 0.792 \\
\quad Horoscope + Age (2-way) & 0.849 & 0.844 & 0.873 \\
\quad \textbf{All three tasks (3-way)} & \textbf{0.864} & \textbf{0.860} & \textbf{0.888} \\
\midrule
\rowcolor{green!10}
\textit{Impr. over Best Single} & \textit{+9.8\%} & \textit{+10.4\%} & \textit{+6.7\%} \\
\bottomrule
\end{tabular}
}
\end{table}

\noindent\textbf{Result and Analysis.} Table~\ref{tab:llm_results} presents and compares MIA performance.
The 3-way \name attack achieves AUC of 0.888, representing a \textbf{6.7\%} relative improvement over the best single-task attack (0.742). This demonstrates that the ODMM privacy composition holds for autoregressive language models beyond discriminative classifiers and regression.

% \mypara{Consistent Amplification Across Task Combinations.}
% All multi-task combinations outperform their constituent single-task attacks. Even the weakest 2 way combination (Gender + Age, AUC = 0.792) exceeds the best single-task attack (Horoscope, AUC = 0.832 is close but the combination still provides balanced improvement). The 3 way fusion consistently achieves the highest performance, confirming that additional tasks provide incremental privacy leakage.

\noindent$\bullet$\textit{Weak Signal Remain Valuable.}
The Gender task exhibits the weakest individual attack performance (AUC = 0.663), yet its inclusion in fusion attacks provides meaningful improvements. Gender + Horoscope (AUC = 0.859) outperforms Horoscope alone (AUC = 0.832), demonstrating that even tasks with weak signals contribute when combined.

% \mypara{Implications for LLM-as-a-Service.}
% These findings carry significant implications for LLM deployment scenarios. Organizations frequently fine-tune the same base model on proprietary data for multiple downstream applications (e.g., sentiment analysis, topic classification, named entity recognition). Our results demonstrate that exposing multiple fine-tuned endpoints or adapters, even through black-box API access, amplifies privacy risks.

\noindent$\bullet$\textit{LLM PEFT Remain Vulnerable.}
Despite using PEFT (i.e., LoRA), which adds only a small fraction of model parameters, the privacy amplification effect remains substantial. This suggests that memorization occurs at the representation level rather than being tied to specific parameter updates, and privacy-preserving techniques must account for multi-task exposure scenarios.

\begin{mdframed}[backgroundcolor=black!10,rightline=false,leftline=false,topline=false,bottomline=false,roundcorner=2mm]
\textbf{Takeaway 2:} In addition to \textbf{Takeaway 1}, first, the ODMM privacy composition is \textit{universal}: \name achieves 6.7\% to 16.6\% AUC improvement regardless of data modality (image vs. text), model family (CNN vs. Transformer vs. LLM), or training paradigm (full fine-tuning vs. Qwen3-1.7B with LoRA). Second, \textit{weak MIA tasks contribute meaningfully}: across all scenarios, tasks with near-random individual attack performance still provide non-redundant membership signals that amplify joint inference, validating the additive leakage accumulation in Theorem~\ref{thm:llr-decomp}. Third, \textit{Heterogeneity/divergence enhances amplification}: UTKFace (heterogeneous tasks) and HAM10000 (heterogeneous architectures) show the strongest gains, consistent with Theorem~\ref{thm:keff}. 
\end{mdframed}

% ======================================================
\subsection{Multi-Tasks vs. Single-Task-Models}
\label{sec:discussion_mtl}
% ======================================================

A natural question arises: \textit{Can multi-task learning (MTL) with a shared backbone mitigate the privacy amplification risk identified in this work?} Prior work by Yan \textit{et al.}~\cite{yan2023mtl} investigates MIA in MTL settings where multiple tasks share an identical backbone with task-specific heads. This architectural choice differs fundamentally from our setting, where each task is served by an independently trained model. 
% We conduct experiments to compare these two paradigms and \textit{reveal a critical privacy-utility trade-off}.

% \subsubsection{Experimental Setup}

\noindent\textbf{Experimental Setup.} We compare two model architectures on UTKFace. \Circled{1} \textit{Multi-Tasks Model}: A single ResNet-18 backbone shared across all three tasks (Age, Gender, Race), with task-specific linear heads attached to the shared feature extractor. This is what Yan \textit{et al.}~\cite{yan2023mtl} investigate. \Circled{2} \textit{Single-Task Models}: Three independently trained ResNet-18 models, each dedicated to a single task with full parameter independence. This is what \name investigates.

Both configurations are trained for 50 epochs using the Adam optimizer with a learning rate $10^{-3}$. 
For the multi-tasks model, we employ a joint training objective that sums the losses across all tasks. 
We evaluate both task performance and MIA vulnerability under identical data splits.

% \subsubsection{Utility Comparison}

\begin{table}[h]
\centering
\caption{Task utility comparison (UTKFace).}
\label{tab:mtl_performance}
\resizebox{0.95\linewidth}{!}{
\begin{tabular}{llccc}
\toprule
\textbf{Task} & \textbf{Metric} & \textbf{Multi-Tasks} & \textbf{Single-Task} & \textbf{$\Delta$} \\
\midrule
Age & MAE $\downarrow$ & 7.30 & 7.22 & -0.08 \\
Gender & Acc $\uparrow$ & 80.89\% & 87.51\% & -6.62\% \\
Race & Acc $\uparrow$ & 62.49\% & 75.40\% & -12.91\% \\
\bottomrule
\end{tabular}
}
\end{table}

% \subsubsection{Privacy Vulnerability Comparison}

\begin{table}[h]
\centering
\caption{MIA comparison.}
\label{tab:mtl_mia}
\resizebox{0.90\linewidth}{!}{
\begin{tabular}{llcc}
\toprule
\textbf{Attack Type} & \textbf{Metric} & \textbf{Multi-Tasks} & \textbf{Single-Task} \\
\midrule
\multirow{3}{*}{Single} 
& Age AUC & 0.778 & 0.789 \\
& Gender AUC & 0.494 & 0.688 \\
& Race AUC & 0.529 & 0.844 \\
\midrule
\multirow{4}{*}{\name}
& Accuracy & 0.730 & \textbf{0.921} \\
& F1 Score & 0.741 & \textbf{0.925} \\
& AUC & 0.799 & \textbf{0.960} \\
& Member Recall & 0.773 & \textbf{0.979} \\
\bottomrule
\end{tabular}
}
\end{table}

\noindent\textbf{Utility Comparison.} 
Table~\ref{tab:mtl_performance} presents the task-specific performance comparison between multi-tasks and single-task architectures. The results reveal a substantial \textit{utility gap} favoring single-task models:

\noindent\textit{Age Regression}: Comparable performance (MAE difference of only 0.08 years), suggesting that age-related features are relatively task-agnostic.

\noindent\textit{Gender Classification}: Single-task models achieve 6.62\% higher accuracy, indicating that gender-specific features benefit from dedicated optimization.

\noindent\textit{Race Classification}: The largest gap of 12.91\% demonstrates significant negative transfer, the shared backbone learns compromised representations when forced to simultaneously optimize for unaligned task objectives.

%Table~\ref{tab:mtl_performance} presents the task-specific utility performance comparison. The results reveal a substantial \textit{utility gap} favoring single-task models (detailed analysis is in \ref{app:multiheadCompare}). 
This performance degradation in \textit{multi-tasks} models is consistent with the well-documented \textit{negative transfer} phenomenon in multi-task learning~\cite{standley2020taskslearnedmultitasklearning, yu2020gradientsurgerymultitasklearning}, where gradient conflicts between tasks lead to suboptimal solutions for individual objectives. To gain high utility, a single task/head model must always be chosen in practice.

\noindent\textbf{Privacy Vulnerability Comparison.} 
%Table~\ref{tab:mtl_mia} presents the MIA performance, revealing that single-task models are substantially more vulnerable to privacy leakage. %For example, the \name attack against single-task models achieves an AUC of 0.960, compared to 0.799 for multi-task models, a striking \textbf{20.1\% absolute improvement} in attack effectiveness.    
Table~\ref{tab:mtl_mia} presents the MIA performance, revealing:

\noindent\textit{Single-task Models are More Vulnerable}: The \name attack against single-task models achieves an AUC of 0.960, compared to 0.799 for multi-tasks models, a striking \textit{20.1\% absolute improvement} in attack effectiveness.
    
\noindent\textit{Task-Specific Vulnerability Amplification}: Single-task attacks against single-task models consistently outperform those against multi-tasks models. The Race task shows the most dramatic difference (AUC: 0.844 vs. 0.529), corresponding to the task where single-task models achieve the largest performance advantage.
    
\noindent\textit{Privacy Composition is Stronger in Independent Models}: The \name over the best single-task attack is more pronounced for single-task models (0.960 vs. 0.844 = +13.7\%) compared to multi-tasks models (0.799 vs. 0.778 = +2.7\%). 
% This confirms that independently trained models develop more diverse and complementary memorization patterns.

\noindent\textbf{Why Single-Task Models are More Vulnerable?}
There are three potential factors:

\noindent\textit{Task-Specific Overfitting}: Each single-task model can fully overfit to its specific task objective without interference, leading to stronger memorization of training samples from the task's perspective.
    
\noindent\textit{Independent Optimization Trajectories}: Without shared parameters, single-task models explore \textit{different regions} of the loss landscape, resulting in complementary (rather than redundant) memorization patterns.

\noindent\textit{Divergent Learning Features}: Because a multi-tasks model shares the same backbone, the backbone is forced to learn representations that are simultaneously useful across all tasks. This constraint limits the learning of task-specific features and is therefore suboptimal from a utility perspective. In contrast, a single-task model can fully specialize in task-specific representations, achieving higher utility but learning features that are divergent across tasks. When such features are collectively exposed through ODMM models, they effectively \emph{expand the memorable feature space}, thereby amplifying privacy leakage.

% In summary, our findings reveal a fundamental \textit{privacy-utility trade-off} between the two architectural paradigms (reasons are explained in detail in Appendix~\ref{app:multiheadCompare}):

\begin{mdframed}[backgroundcolor=black!10,rightline=false,leftline=false,topline=false,bottomline=false,roundcorner=2mm]
\textbf{Key Finding:} There is a trade-off. A multi-task model with a shared backbone is less vulnerable to \textit{privacy protection} through representation coupling, but at the cost of \textit{significant utility (e.g., accuracy) degradation}. Single-task models achieve superior utility but create amplified privacy risks when jointly exposed.
\end{mdframed}

% \noindent\textbf{Higher Task Performance Correlates with Higher Memorization}: The superior task performance of single-head models (Table~\ref{tab:mtl_performance}) indicates stronger fitting to training data, which inherently increases privacy leakage.

% ======================================================
% \subsection{Ablation Study}
% \label{sec:ablation}
% ======================================================

% We conduct ablation studies to investigate key factors affecting the privacy amplification phenomenon. 

%, and the impact of training data overlap rates across tasks (Section~\ref{sec:ablation_overlap}).

% ======================================================
\subsection{Number of Exposed Models}
\label{sec:ablation_num_models}
% ======================================================

% A natural question arises: \textit{how does privacy risk scale with the number of multi-task models exposed?} We hypothesize that additional models provide incremental privacy leakage due to our Memory Complementarity principle, but the marginal gain may diminish as redundancy increases.

We conduct this study on CelebA using ResNet-18 models trained for 50 epochs. We progressively expand the set of ODMM tasks from 3 to 7, selecting attributes that represent diverse facial characteristics: \textbf{3 Models} of Male, Mouth\_Slightly\_Open, Big\_Nose, \textbf{4 Models} by + Black\_Hair, \textbf{5 Models} by + Smiling, \textbf{6 Models} by + Straight\_Hair.
%and \textbf{7 Models} by + Wearing\_Hat
Each configuration uses the same victim/shadow data splits to ensure fair comparison. The joint feature dimension scales linearly: $3\times k$ dimensions for $k$ models.

\begin{figure}[h]
    \centering
    \includegraphics[width=0.75\linewidth]{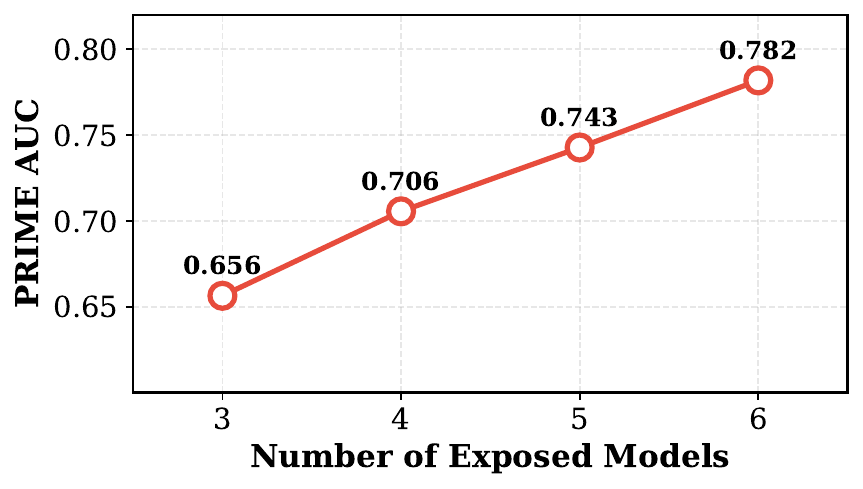}
    \caption{MIA AUC vs number of ODMM models.}
    \label{fig:num_models_ablation}
\end{figure}

% Figure~\ref{fig:num_models_ablation} and Table~\ref{tab:num_models_ablation} 
Figure~\ref{fig:num_models_ablation}  presents the \name performance as the number of exposed models increases. The \name AUC increases monotonically from 0.656 (3 models) to 0.782 (6 models), representing a \textbf{12.6\% improvement}.
The incremental AUC gain per model shows diminishing returns: +0.050 (3$\rightarrow$4), +0.037 (4$\rightarrow$5), +0.039 (5$\rightarrow$6).
This confirms that each additional task-specific model contributes unique membership signals, confirming Theorem~\ref{thm:llr-decomp}.

\section{Discussion}
\label{sec:discussion}
% ======================================================

We discuss three important aspects that adhere to our study: 
%(1) the comparison between multi-task learning with shared backbones and independent single-task models; 
(1) evaluating two potential defenses; (2) using the SOTA RMIA~\cite{ZarifzadehLS24} method to show that \name consistently has superior performance over single-task, regardless of its underlying MIA method.

\subsection{Differential Privacy as a Defense}
\label{sec:dp_defense}

% Differential Privacy (DP)~\cite{dwork2006calibrating} provides mathematically rigorous privacy guarantees by ensuring that the model's output distribution is approximately independent of any single training sample. 
We evaluate whether DP-SGD~\cite{abadi2016deep}, the standard approach for training DL models with DP~\cite{dwork2006calibrating}, can effectively mitigate \name attack. We conduct DP experiments on UTKFace using a compact CNN architecture compatible with Opacus~\cite{yousefpour2021opacus} using the DP-SGD framework. The detailed experimental setup is in Appendix~\ref{app:dpDefense}.

\begin{table}[h]
\centering
\caption{Model utility under DP (UTKFace). 
% $\varepsilon$ is the privacy budget consumed by 50 epochs of training.
}
\label{tab:dp_utility}
\resizebox{1.0\linewidth}{!}{
\begin{tabular}{lcccccc}
\toprule
\textbf{Configuration} & $\boldsymbol{\sigma}$ & $\boldsymbol{\varepsilon}$ & \textbf{Age MAE}$\downarrow$ & \textbf{Gender Acc}$\uparrow$ & \textbf{Race Acc}$\uparrow$ \\
\midrule
No DP & -- & $\infty$ & 7.47 & 89.16\% & 76.58\% \\
Moderate DP & 0.2 & 635.2 & 8.72 & 86.71\% & 71.98\% \\
Strong DP & 0.5 & 22.7 & 10.20 & 86.03\% & 67.00\% \\
\bottomrule
\end{tabular}
}
\end{table}

\noindent\textbf{Key Result.} Comprehensive result interpretations are in Appendix~\ref{app:dpDefense}, while we provide key results herein. Table~\ref{tab:dp_utility} presents the task performance under different DP configurations. As expected, adding noise during training degrades model utility across all tasks. Table~\ref{tab:dp_attack} shows the MIA performance under DP protection. The results demonstrate that DP-SGD provides substantial defense against both single-task and \name joint attacks. To quantify the defense effectiveness, we define the \textit{Protection Rate} as the fraction of excess attack AUC (above the random baseline of 0.5) that is eliminated:
\begin{equation}
\text{Protection Rate} = \frac{\text{AUC}_{\text{no-DP}} - \text{AUC}_{\text{DP}}}{\text{AUC}_{\text{no-DP}} - 0.5} \times 100\%.
\end{equation}
Table~\ref{tab:dp_tradeoff} summarizes the privacy-utility trade-off, where we use Race classification accuracy loss as a representative utility metric.

\begin{table}[t]
\centering
\caption{MIA performance under DP protection.}
\label{tab:dp_attack}
\resizebox{1.0\linewidth}{!}{
\begin{tabular}{lccccc}
\toprule
\textbf{Configuration} & \textbf{Age AUC} & \textbf{Gender AUC} & \textbf{Race AUC} & \textbf{Avg. Single} & \textbf{\name (Joint)} \\
\midrule
No DP & 0.740 & 0.567 & 0.717 & 0.674 & 0.879 \\
Moderate DP ($\sigma$=0.2) & 0.533 & 0.524 & 0.526 & 0.528 & 0.550 \\
Strong DP ($\sigma$=0.5) & 0.514 & 0.512 & 0.509 & 0.511 & 0.519 \\
\bottomrule
\end{tabular}
}
\end{table}

\begin{table}[t]
\centering
\caption{Privacy-utility trade-off for DP defense. 
% Protection Rate measures defense effectiveness against \name; Utility Loss denotes Race classification accuracy degradation.
}
\label{tab:dp_tradeoff}
\resizebox{1.0\linewidth}{!}{
\begin{tabular}{lccccc}
\toprule
\textbf{Configuration} & \textbf{\name AUC} & \textbf{AUC Drop} & \textbf{Protection Rate} & \textbf{Race Acc Loss} \\
\midrule
Moderate DP ($\sigma$=0.2) & 0.550 & 0.329 & 86.7\% & 4.60\% \\
Strong DP ($\sigma$=0.5) & 0.519 & 0.360 & 95.1\% & 9.58\% \\
\bottomrule
\end{tabular}
}
\end{table}

% In summary, DP as a defense suffers a notable privacy utility trade-off, preventing its practical usage as a mitigation.

\begin{mdframed}[backgroundcolor=black!10,rightline=false,leftline=false,topline=false,bottomline=false,roundcorner=2mm]
\textbf{Key Finding}: DP is an effective defense against \name but incurs a notable utility trade-off. Even under DP protection, \name consistently outperforms single-task MIAs, confirming the persistence of ODMM privacy composition.
\end{mdframed}

% ======================================================
\subsection{Limiting Information: Hard-Label}
\label{sec:hard_label}
% ======================================================

As another defense, model providers can restrict API outputs to hard labels (e.g., class predictions or scalar values) rather than full probability distributions~\cite{li2021membership}. We now evaluate \name that operates solely on hard labels.

\mypara{Methodology.}
We employ a perturbation-based approach~\cite{li2021membership, choquette2021label} to extract membership signals. The intuition is that member samples are generally more robust to input noise than non-members. We generate $N=30$ perturbed copies of a target image by adding Gaussian noise and query the ODMM models. For each task, we extract a 4-dimensional feature vector based on the model's discrete outputs: (1) prediction correctness, (2) consistency ratio (stability of prediction under noise), (3) correctness ratio under noise, and (4) prediction entropy across perturbations. For the regression task (Age), we discretize the scalar output into 10-year bins to calculate consistency metrics. More details are in Appendix~\ref{app:hard_label}.

\mypara{Results.}
We evaluate on UTKFace. Table~\ref{tab:hard_label_summary} summarizes the results, including the original victim model utility.

\begin{table}[h]
\centering
\caption{MIA performance under hard-label on UTKFace.}
\label{tab:hard_label_summary}
\resizebox{0.8\linewidth}{!}{
\begin{tabular}{lcccc}
\toprule
\textbf{Attack Configuration} & \textbf{Utility} & \textbf{MIA Acc} & \textbf{MIA F1} & \textbf{MIA AUC} \\
\midrule
Single-Age (Hard) & 5.93 (MAE) & 0.675 & 0.733 & 0.694 \\
Single-Gender (Hard) & 91.18\% & 0.558 & 0.661 & 0.565 \\
Single-Race (Hard) & 79.70\% & 0.596 & 0.708 & 0.602 \\
\midrule\rowcolor{green!10}
\textit{Impr. over Best Single} & -- & \textit{+10.1\%} 
& \textit{+5.5\%} & \textit{+11.7\%} \\
\bottomrule
\end{tabular}
}
\end{table}

The results show that ODMM privacy risks persist even without access to confidence scores. While single-task hard-label attacks are relatively weak (average AUC 0.620), the joint \name attack achieves an AUC of \textbf{0.775}, representing a \textbf{24.9\%} improvement.
Feature importance analysis reveals that Age (46.3\%) and Race (35.7\%) tasks contribute the most significant robustness signals, while Gender (18.0\%) contributes less.
% Therefore, limiting exposure information to hard-label cannot eliminate \name as the \name shows \textit{consistent substantial improved attack effectiveness} compared to single-task attack, though it can reduce the \textit{absolute} MIA accuracy.
% This confirms that the \textit{complementary robustness} of a sample across different tasks can be exploited to infer membership, rendering simple confidence-masking defenses insufficient for ODMM deployments.

% ======================================================

\begin{mdframed}[backgroundcolor=black!10,rightline=false,leftline=false,topline=false,bottomline=false,roundcorner=2mm]
\textbf{Key Finding:} The privacy amplification effect of \name is not specific to the MIA method. It always outperforms single-task attacks across both traditional shadow-model-based MIA and the SOTA RMIA regardless of accessible information (e.g., hard-label only).
% , confirming that the fundamental vulnerability lies in the \textit{complementary memorization} induced by ODMM training.
\end{mdframed}

\subsection{Extend RMIA to \name}
%\subsection{Robust Membership Inference Attack}
\label{sec:rmia_validation}
% ======================================================

Our experiments have employed the classical MIA paradigm, where a binary attack model is trained (using shadow models) to distinguish members from non-members. A natural question arises: \textit{Does the privacy amplification effect of \name generalize to SOTA MIA methods?} To address this, we evaluate \name using RMIA~\cite{ZarifzadehLS24}, an SOTA attack that achieves superior performance, particularly in the low false-positive rate (FPR) regime critical for practical privacy auditing.

%\subsubsection{RMIA Experiment Details}
\label{app:rmia_details}

\subsubsection{RMIA Methodology}

RMIA~\cite{ZarifzadehLS24} employs pairwise likelihood ratio testing instead of training attack classifiers. For a target sample $x$ and target model $\theta$, RMIA computes:
\begin{equation}
\text{Score}_{\text{RMIA}}(x) = \Pr_{z \sim \pi}\left[\frac{\Pr(x|\theta)/\Pr(x)}{\Pr(z|\theta)/\Pr(z)} \geq \gamma\right],
\end{equation}
where $\Pr(x|\theta)$ is the prediction score (softmax) of the model's output for the true-label class, $\Pr(x)$ is the baseline prediction score estimated from reference models, and $\gamma$ is a threshold parameter. The $\text{Score}_{\text{RMIA}}$ represents the fraction of population samples $z$ that $x$ dominates in terms of likelihood ratio.

Following~\cite{ZarifzadehLS24}, we employ the offline mode with linear correction:
\begin{equation}
\Pr(x) \approx \frac{1}{2}\left[(1 + a) \cdot \Pr_{\text{OUT}}(x) + (1 - a)\right],
\end{equation}
where $\Pr_{\text{OUT}}(x)$ is the average confidence across OUT reference models and $a$ is a correction parameter.

\subsubsection{Joint RMIA for ODMM}

To extend RMIA to \name and fit the ODMM setting, we joint likelihood ratios across task-specific models using geometric mean:
\begin{equation}
\text{Ratio}_{\text{joint}}(x) = \left(\prod_{k=1}^{K} \frac{\Pr(x|\theta_k)}{\Pr(x)}\right)^{1/K}.
\end{equation}
This design preserves the multiplicative nature of likelihood ratios while normalizing across tasks with different ratio magnitudes. We then compute the final pairwise joint score by comparing the joint ratio against population samples. Moreover, we employ binary search on pre-sorted population ratios for efficient computation, reducing per-sample complexity from $O(|Z|)$ to $O(\log|Z|)$.

\subsubsection{Experimental Configuration}

\mypara{Dataset and Models.}
We use CelebA with three binary classification tasks (Male, Mouth\_Slightly\_Open, Big\_Nose). All models are ResNet-18 trained for 50 epochs with Adam optimizer (learning rate $10^{-3}$, batch size 64), using the same data splits as Section~\ref{sec:eval_face} with $N$=5 shadow models per task.

\mypara{Parameters.}
Following~\cite{ZarifzadehLS24}, we set the correction parameter $a = 0.5$, the threshold $\gamma = 2.0$, and the population size $|Z| \approx 40,000$.

\mypara{Confidence Computation for Binary Classification.}
For a sample $(x, y)$ with $y \in \{0, 1\}$:
\begin{equation}
\Pr(x|\theta) = 
\begin{cases}
\sigma(f_\theta(x)), & \text{if } y = 1 \\
1 - \sigma(f_\theta(x)), & \text{if } y = 0
\end{cases}
\end{equation}
where $\sigma(\cdot)$ is the sigmoid function.

% Detailed methodology and experimental setup are provided in Appendix~\ref{app:rmia_details}.

\subsubsection{Results}

Table~\ref{tab:rmia_detailed} presents complete single-task results. RMIA achieves a higher AUC than traditional MIA across all tasks (+4.0\% average), with a particularly pronounced advantage in the low-FPR regime: at TPR@1\%FPR, RMIA reaches 4.4\% on average versus only 1.4\% for traditional MIA. The results yield several important observations:

\begin{table}[h]
\centering
\caption{Single-task MIA performance on CelebA.}
\label{tab:rmia_detailed}
\resizebox{\linewidth}{!}{
\begin{tabular}{llccccccc}
\toprule
\multirow{2}{*}{\textbf{Method}} &
\multirow{2}{*}{\textbf{Task}} &
\multirow{2}{*}{\textbf{AUC}} &
\multirow{2}{*}{\textbf{Accuracy}} &
\multirow{2}{*}{\textbf{Precision}} &
\multirow{2}{*}{\textbf{Recall}} &
\multirow{2}{*}{\textbf{F1}} &
\multicolumn{2}{c}{\textbf{TPR@FPR}} \\
\cmidrule(lr){8-9}
&&&&&&&
\textbf{1\%} &
\textbf{0.1\%} \\
\midrule

\multirow{4}{*}{Traditional}
& Male         & 0.521 & 0.517 & 0.511 & 0.750 & 0.608 & 0.010 & 0.001 \\
& Mouth\_Open  & 0.544 & 0.542 & 0.525 & 0.905 & 0.664 & 0.012 & 0.001 \\
& Big\_Nose    & 0.649 & 0.617 & 0.574 & 0.904 & 0.702 & 0.020 & 0.002 \\
& \textit{Average}
               & \textit{0.571} & \textit{0.544} & \textit{0.541}
               & \textit{0.585} & \textit{0.562}
               & \textit{0.014} & \textit{0.001} \\
\midrule

\multirow{4}{*}{RMIA}
& Male         & 0.533 & 0.517 & 0.547 & 0.193 & 0.286 & 0.024 & 0.005 \\
& Mouth\_Open  & 0.579 & 0.541 & 0.587 & 0.276 & 0.376 & 0.047 & 0.015 \\
& Big\_Nose    & 0.670 & 0.604 & 0.590 & 0.683 & 0.633 & 0.060 & 0.018 \\
& \textit{Average}
               & \textit{0.594} & \textit{0.554} & \textit{0.575}
               & \textit{0.384} & \textit{0.432}
               & \textit{0.044} & \textit{0.013} \\
\bottomrule
\end{tabular}
}
\end{table}

% \begin{table}[h]
% \centering
% \caption{Single-task MIA performance on CelebA.}
% \label{tab:rmia_detailed}
% \resizebox{\linewidth}{!}{
% \begin{tabular}{llccccccc}
% \toprule
% \textbf{Method} & \textbf{Task} & \textbf{AUC} & \textbf{Accuracy} & \textbf{Precision} & \textbf{Recall} & \textbf{F1} & \textbf{TPR@1\%FPR} & \textbf{TPR@0.1\%FPR} \\
% \midrule
% \multirow{4}{*}{Traditional}
% & Male & 0.521 & 0.517 & 0.511 & 0.750 & 0.608 & 0.010 & 0.001 \\
% & Mouth\_Open & 0.544 & 0.542 & 0.525 & 0.905 & 0.664 & 0.012 & 0.001 \\
% & Big\_Nose & 0.649 & 0.617 & 0.574 & 0.904 & 0.702 & 0.020 & 0.002 \\
% & \textit{Average} & \textit{0.571} & \textit{0.544} & \textit{0.541} & \textit{0.585} & \textit{0.562} & \textit{0.014} & \textit{0.001} \\
% \midrule
% \multirow{4}{*}{RMIA}
% & Male & 0.533 & 0.517 & 0.547 & 0.193 & 0.286 & 0.024 & 0.005 \\
% & Mouth\_Open & 0.579 & 0.541 & 0.587 & 0.276 & 0.376 & 0.047 & 0.015 \\
% & Big\_Nose & 0.670 & 0.604 & 0.590 & 0.683 & 0.633 & 0.060 & 0.018 \\
% & \textit{Average} & \textit{0.594} & \textit{0.554} & \textit{0.575} & \textit{0.384} & \textit{0.432} & \textit{0.044} & \textit{0.013} \\
% \bottomrule
% \end{tabular}
% }
% \end{table}

Table~\ref{tab:rmia_summary} presents a comparison of attack performance between single-task and ODMM (via \name), under both traditional MIA and RMIA paradigms on CelebA.

\begin{table}[h]
\centering
\caption{Comparison of \name effectiveness across MIA paradigms on CelebA (RMIA method).}
\label{tab:rmia_summary}
\resizebox{\linewidth}{!}{
\begin{tabular}{llcccccc}
\toprule
\multirow{2}{*}{\textbf{MIA Method}} &
\multirow{2}{*}{\textbf{Attack Type}} &
\multirow{2}{*}{\textbf{AUC}} &
\multirow{2}{*}{\textbf{Accuracy}} &
\multirow{2}{*}{\textbf{F1}} &
\multicolumn{2}{c}{\textbf{TPR@FPR}} \\
\cmidrule(l){6-7}
&&&&&
\textbf{1\%} &
\textbf{0.1\%} \\
\midrule

\multirow{3}{*}{Traditional}
& Avg. Single-Task
    & 0.571 & 0.544 & 0.562 & 0.014 & 0.001 \\
& \textbf{\name (Joint)}
    & \textbf{0.732} & \textbf{0.670} & \textbf{0.721} & \textbf{0.047} & \textbf{0.008} \\
& \textit{Improvement}
    & \textit{+16.1\%} & \textit{+12.6\%} & \textit{+15.9\%} & \textit{+3.3\%} & \textit{+0.7\%} \\
\midrule

\multirow{3}{*}{RMIA}
& Avg. Single-Task
    & 0.594 & 0.554 & 0.432 & 0.044 & 0.013 \\
& \textbf{\name (Joint)}
    & \textbf{0.722} & \textbf{0.626} & \textbf{0.727} & \textbf{0.087} & \textbf{0.017} \\
& \textit{Improvement}
    & \textit{+12.7\%} & \textit{+9.1\%} & \textit{+29.5\%} & \textit{+4.3\%} & \textit{+0.4\%} \\
\bottomrule
\end{tabular}
}
\end{table}

% \begin{table}[h]
% \centering
% \caption{Comparison of \name effectiveness across MIA paradigms on CelebA (RMIA method).}
% \label{tab:rmia_summary}
% \resizebox{\linewidth}{!}{
% \begin{tabular}{llccccc}
% \toprule
% \textbf{MIA Method} & \textbf{Attack Type} & \textbf{AUC} & \textbf{Accuracy} & \textbf{F1} & \textbf{TPR@1\%FPR} & \textbf{TPR@0.1\%FPR} \\
% \midrule
% \multirow{3}{*}{Traditional} 
% & Avg. Single-Task & 0.571 & 0.544 & 0.562 & 0.014 & 0.001 \\
% & \textbf{\name (Joint)} & \textbf{0.732} & \textbf{0.670} & \textbf{0.721} & \textbf{0.047} & \textbf{0.008} \\
% & \textit{Improvement} & \textit{+16.1\%} & \textit{+12.6\%} & \textit{+15.9\%} & \textit{+3.3\%} & \textit{+0.7\%} \\
% \midrule
% \multirow{3}{*}{RMIA} 
% & Avg. Single-Task & 0.594 & 0.554 & 0.432 & 0.044 & 0.013 \\
% & \textbf{\name (Joint)} & \textbf{0.722} & \textbf{0.626} & \textbf{0.727} & \textbf{0.087} & \textbf{0.017} \\
% & \textit{Improvement} & \textit{+12.7\%} & \textit{+9.1\%} & \textit{+29.5\%} & \textit{+4.3\%} & \textit{+0.4\%} \\
% \bottomrule
% \end{tabular}
% }
% \end{table}

\noindent$\bullet$\textbf{\name Extends Effectively to RMIA}. The \name under the RMIA paradigm achieves an AUC of 0.722, substantially outperforming the average single-task RMIA (AUC=0.594) by \textbf{12.7\%}. It also attains \textbf{4.3\% } higher TPR @ 1\% FPR compared to single-task RMIA, demonstrating \name's privacy amplification effect in the critical TPR@Low FPR regime.

\noindent$\bullet$\textbf{RMIA Strengthens Low-FPR Detection.} In the low-FPR regime, single-task RMIA detects 4.4\% of members at 1\% FPR (versus only 1.4\% for traditional single-task MIA), and \name further amplifies this to \textbf{8.7\%}. Notably, although \name's overall AUC under RMIA (0.722) is comparable to that under the traditional paradigm (0.732), its TPR@1\%FPR (8.7\%) nearly doubles that of the traditional joint attack (4.7\%), consistent with RMIA's design objective of maximizing detection at low false-positive rates rather than overall AUC.

\noindent$\bullet$\textbf{Consistent Privacy Amplification Across Paradigms.} The AUC performance gain from \name remains substantial under both traditional MIA (+16.1\%) and RMIA (+12.7\%) paradigms, reinforcing evidence ODMM privacy composition.

%======================================================
\section{Conclusion}
\label{sec:con}

For the first time, this work theoretically and empirically reveals a striking form of privacy leakage rooted in ODMM privacy composition, a practice increasingly adopted to maximize the value of high-quality data. Grounded in our established theoretical foundation, we propose \name, a principled framework for evaluating privacy risks in the ODMM paradigm. Comprehensive experiments across diverse privacy-sensitive datasets, multiple model architectures, and a wide range of application scenarios consistently validate our root-cause hypothesis of \emph{ODMM privacy composition}. We hope this work raises awareness of this potent and previously unrecognized risk, and encourages the community to exercise greater caution when leveraging valuable data across multiple learning objectives.
% ======================================================

%This work presents the first in-depth privacy evaluation framework \name for three widely used and commercially supported compression operations---pruning, quantization, and weight clustering---through the lens of membership inference. Specifically, \nameappend$_{\rm SR}$ reveals that these compression operations indeed increase privacy leakage. Notably, this leakage is further exacerbated when leveraging information from multiple compressed models.  Building on this, we propose \nameappend$_{\rm MR}$ that stacks the loss of multiple compressed models and the meta-posterior from the \nameappend$_{\rm SR}$ attack meta-classifier. Extensive experiments have validated the \name superior performance under diverse model architectures, datasets, and various settings.

%\section*{Acknowledgments}
%-------------------------------------------------------------------------------

%\textbf{Do not include any acknowledgements in your submission which may deanonymize you (e.g., because of specific affiliations or grants you acknowledge)}

%-------------------------------------------------------------------------------
% optional clearing of the page
\cleardoublepage
%\appendix
% optional clearing of the page
% \cleardoublepage

% optional clearing of the page
% \cleardoublepage
%\section*{Acknowledgment}

\bibliographystyle{unsrt}
\bibliography{normal_generated_py3}

@misc{chatgpt,
  howpublished = {\url{https://openai.com/chatgpt}}
}

@misc{midjourney,
  howpublished = {\url{https://www.midjourney.com}}
}

@misc{detectron2,
  howpublished = {\url{https://github.com/facebookresearch/detectron2/blob/main/MODEL_ZOO.md}}
}

@inproceedings{hu2023planning,
  title={Planning-oriented autonomous driving},
  author={Hu, Yihan and Yang, Jiazhi and Chen, Li and Li, Keyu and Sima, Chonghao and Zhu, Xizhou and Chai, Siqi and Du, Senyao and Lin, Tianwei and Wang, Wenhai and others},
  booktitle={Proceedings of the IEEE/CVF Conference on Computer Vision and Pattern Recognition},
  pages={17853--17862},
  year={2023}
}

@inproceedings{deng2019arcface,
  title={Arcface: Additive angular margin loss for deep face recognition},
  author={Deng, Jiankang and Guo, Jia and Xue, Niannan and Zafeiriou, Stefanos},
  booktitle={Proceedings of the IEEE/CVF Conference on Computer Vision and Pattern Recognition},
  pages={4690--4699},
  year={2019}
}

@inproceedings{ma2024watch,
  title={Watch out! simple horizontal class backdoor can trivially evade defense},
  author={Ma, Hua and Wang, Shang and Gao, Yansong and Zhang, Zhi and Qiu, Huming and Xue, Minhui and Abuadbba, Alsharif and Fu, Anmin and Nepal, Surya and Abbott, Derek},
  booktitle={Proceedings of the 2024 on ACM SIGSAC Conference on Computer and Communications Security},
  pages={4465--4479},
  year={2024}
}

@inproceedings{li2024yes,
  title={Yes,$\{$One-Bit-Flip$\}$ Matters! Universal $\{$DNN$\}$ Model Inference Depletion with Runtime Code Fault Injection},
  author={Li, Shaofeng and Wang, Xinyu and Xue, Minhui and Zhu, Haojin and Zhang, Zhi and Gao, Yansong and Wu, Wen and Shen, Xuemin Sherman},
  booktitle={33rd USENIX Security Symposium (USENIX Security 24)},
  pages={1315--1330},
  year={2024}
}

@article{wang2024greats,
  title={Greats: Online selection of high-quality data for llm training in every iteration},
  author={Wang, Jiachen T and Wu, Tong and Song, Dawn and Mittal, Prateek and Jia, Ruoxi},
  journal={Advances in Neural Information Processing Systems},
  volume={37},
  pages={131197--131223},
  year={2024}
}

@inproceedings{li2024quantity,
  title={From quantity to quality: Boosting llm performance with self-guided data selection for instruction tuning},
  author={Li, Ming and Zhang, Yong and Li, Zhitao and Chen, Jiuhai and Chen, Lichang and Cheng, Ning and Wang, Jianzong and Zhou, Tianyi and Xiao, Jing},
  booktitle={Proceedings of the 2024 Conference of the North American Chapter of the Association for Computational Linguistics: Human Language Technologies (Volume 1: Long Papers)},
  pages={7602--7635},
  year={2024}
}

@article{litjens2017survey,
  title={A survey on deep learning in medical image analysis},
  author={Litjens, Geert and Kooi, Thijs and Bejnordi, Babak Ehteshami and Setio, Arnaud Arindra Adiyoso and Ciompi, Francesco and Ghafoorian, Mohsen and Van Der Laak, Jeroen Awm and Van Ginneken, Bram and S{\'a}nchez, Clara I},
  journal={Medical Image Analysis},
  volume={42},
  pages={60--88},
  year={2017}
}

@article{vandenhende2021multi,
  title={Multi-task learning for dense prediction tasks: A survey},
  author={Vandenhende, Simon and Georgoulis, Stamatios and Van Gansbeke, Wouter and Proesmans, Marc and Dai, Dengxin and Van Gool, Luc},
  journal={IEEE Transactions on Pattern Analysis and Machine Intelligence},
  volume={44},
  number={7},
  pages={3614--3633},
  year={2021},
  publisher={IEEE}
}

@inproceedings{zamir2018taskonomy,
  title={Taskonomy: Disentangling task transfer learning},
  author={Zamir, Amir R and Sax, Alexander and Shen, William and Guibas, Leonidas J and Malik, Jitendra and Savarese, Silvio},
  booktitle={Proceedings of the IEEE Conference on Computer Vision and Pattern Recognition},
  pages={3712--3722},
  year={2018}
}

@inproceedings{ye2022inverted,
  title={Inverted pyramid multi-task transformer for dense scene understanding},
  author={Ye, Hanrong and Xu, Dan},
  booktitle={European Conference on Computer Vision},
  pages={514--530},
  year={2022},
  organization={Springer}
}

@inproceedings{li2022auditing,
  title={Auditing membership leakages of multi-exit networks},
  author={Li, Zheng and Liu, Yiyong and He, Xinlei and Yu, Ning and Backes, Michael and Zhang, Yang},
  booktitle={Proceedings of the 2022 ACM SIGSAC Conference on Computer and Communications Security},
  pages={1917--1931},
  year={2022}
}

@inproceedings{baluta2022membership,
  title={Membership inference attacks and generalization: A causal perspective},
  author={Baluta, Teodora and Shen, Shiqi and Hitarth, S and Tople, Shruti and Saxena, Prateek},
  booktitle={Proceedings of the 2022 ACM SIGSAC Conference on Computer and Communications Security},
  pages={249--262},
  year={2022}
}

@inproceedings{ye2022enhanced,
  title={Enhanced membership inference attacks against machine learning models},
  author={Ye, Jiayuan and Maddi, Aadyaa and Murakonda, Sasi Kumar and Bindschaedler, Vincent and Shokri, Reza},
  booktitle={Proceedings of the 2022 ACM SIGSAC Conference on Computer and Communications Security},
  pages={3093--3106},
  year={2022}
}

@inproceedings{li2021membership,
  title={Membership inference attacks and defenses in classification models},
  author={Li, Jiacheng and Li, Ninghui and Ribeiro, Bruno},
  booktitle={Proceedings of the Eleventh ACM Conference on Data and Application Security and Privacy},
  pages={5--16},
  year={2021}
}

@inproceedings{shokri2017membership,
  title={Membership inference attacks against machine learning models},
  author={Shokri, Reza and Stronati, Marco and Song, Congzheng and Shmatikov, Vitaly},
  booktitle={IEEE Symposium on Security and Privacy (S\&P)},
  pages={3--18},
  year={2017},
  organization={IEEE}
}

@inproceedings{carlini2022membership,
  title={Membership inference attacks from first principles},
  author={Carlini, Nicholas and Chien, Steve and Nasr, Milad and Song, Shuang and Terzis, Andreas and Tramer, Florian},
  booktitle={IEEE Symposium on Security and Privacy (S\&P)},
  pages={1897--1914},
  year={2022},
  organization={IEEE}
}

@inproceedings{li2024seqmia,
  title={Seqmia: Sequential-metric based membership inference attack},
  author={Li, Hao and Li, Zheng and Wu, Siyuan and Hu, Chengrui and Ye, Yutong and Zhang, Min and Feng, Dengguo and Zhang, Yang},
  booktitle={Proceedings of ACM Conference on Computer and Communications Security},
  pages={3496--3510},
  year={2024}
}

@inproceedings{he2024difficulty,
  title={Is Difficulty Calibration All We Need? Towards More Practical Membership Inference Attacks},
  author={He, Yu and Li, Boheng and Wang, Yao and Yang, Mengda and Wang, Juan and Hu, Hongxin and Zhao, Xingyu},
  booktitle={Proceedings of the 2024 on ACM SIGSAC Conference on Computer and Communications Security},
  pages={1226--1240},
  year={2024}
}

@inproceedings{chen2020gan,
  title={Gan-leaks: A taxonomy of membership inference attacks against generative models},
  author={Chen, Dingfan and Yu, Ning and Zhang, Yang and Fritz, Mario},
  booktitle={Proceedings of the 2020 ACM SIGSAC Conference on Computer and Communications Security},
  pages={343--362},
  year={2020}
}

@inproceedings{yeom2018privacy,
  title={Privacy risk in machine learning: Analyzing the connection to overfitting},
  author={Yeom, Samuel and Giacomelli, Irene and Fredrikson, Matt and Jha, Somesh},
  booktitle={2018 IEEE 31st Computer Security Foundations Symposium (CSF)},
  pages={268--282},
  year={2018},
  organization={IEEE}
}

@misc{gdpr,
  howpublished = {\url{https://gdpr-info.eu/.}},
}

@inproceedings{chen2024slmia,
  author       = {Guangke Chen and
                  Yedi Zhang and
                  Fu Song},
  title        = {{SLMIA-SR:} Speaker-Level Membership Inference Attacks against Speaker
                  Recognition Systems},
  booktitle    = {31st Annual Network and Distributed System Security Symposium, {NDSS}
                  2024, San Diego, California, USA, February 26 - March 1, 2024},
  publisher    = {The Internet Society},
  year         = {2024},
}

@inproceedings{meeus2024did,
  title={Did the neurons read your book? document-level membership inference for large language models},
  author={Meeus, Matthieu and Jain, Shubham and Rei, Marek and de Montjoye, Yves-Alexandre},
  booktitle={33rd USENIX Security Symposium (USENIX Security 24)},
  pages={2369--2385},
  year={2024}
}

@inproceedings{ZarifzadehLS24,
  author       = {Sajjad Zarifzadeh and
                  Philippe Liu and
                  Reza Shokri},
  title        = {Low-Cost High-Power Membership Inference Attacks},
  booktitle    = {Forty-first International Conference on Machine Learning, {ICML} 2024,
                  Vienna, Austria, July 21-27, 2024},
  year         = {2024}
}

@inproceedings{HuiYYBGC21,
  author       = {Bo Hui and
                  Yuchen Yang and
                  Haolin Yuan and
                  Philippe Burlina and
                  Neil Zhenqiang Gong and
                  Yinzhi Cao},
  title        = {Practical Blind Membership Inference Attack via Differential Comparisons},
  booktitle    = {28th Annual Network and Distributed System Security Symposium, {NDSS}
                  2021, virtually, February 21-25, 2021},
  year         = {2021}
}

@inproceedings{yuan2022membership,
  title={Membership inference attacks and defenses in neural network pruning},
  author={Yuan, Xiaoyong and Zhang, Lan},
  booktitle={31st USENIX Security Symposium (USENIX Security 22)},
  pages={4561--4578},
  year={2022}
}

@inproceedings{du2026cascading,
  title={Cascading and Proxy Membership Inference Attacks},
  author={Du, Yuntao and Li, Jiacheng and Chen, Yuetian and Zhang, Kaiyuan and Yuan, Zhizhen and Xiao, Hanshen and Ribeiro, Bruno and Li, Ninghui},
  booktitle={33nd Annual Network and Distributed System Security Symposium, {NDSS} 2026},
  year={2026}
}

@inproceedings{Salem0HBF019,
  author       = {Ahmed Salem and
                  Yang Zhang and
                  Mathias Humbert and
                  Pascal Berrang and
                  Mario Fritz and
                  Michael Backes},
  title        = {ML-Leaks: Model and Data Independent Membership Inference Attacks
                  and Defenses on Machine Learning Models},
  booktitle    = {26th Annual Network and Distributed System Security Symposium, {NDSS}
                  2019, San Diego, California, USA, February 24-27, 2019},
  year         = {2019}
}

@inproceedings{song2021systematic,
  title={Systematic evaluation of privacy risks of machine learning models},
  author={Song, Liwei and Mittal, Prateek},
  booktitle={30th USENIX Security Symposium (USENIX Security 21)},
  pages={2615--2632},
  year={2021}
}

@inproceedings{NasehPSCOH25,
  author       = {Ali Naseh and
                  Yuefeng Peng and
                  Anshuman Suri and
                  Harsh Chaudhari and
                  Alina Oprea and
                  Amir Houmansadr},
  title        = {Riddle Me This! Stealthy Membership Inference for Retrieval-Augmented
                  Generation},
  booktitle    = {Proceedings of the 2025 {ACM} {SIGSAC} Conference on Computer and
                  Communications Security, {CCS} 2025, Taipei, Taiwan, October 13-17,
                  2025},
  pages        = {1245--1259},
  publisher    = {{ACM}},
  year         = {2025}
}

@inproceedings{GaoMD0G25,
  author       = {Xinyu Gao and
                  Xiangtao Meng and
                  Yingkai Dong and
                  Zheng Li and
                  Shanqing Guo},
  title        = {\emph{DCMI: } {A} Differential Calibration Membership Inference Attack
                  Against Retrieval-Augmented Generation},
  booktitle    = {Proceedings of the 2025 {ACM} {SIGSAC} Conference on Computer and
                  Communications Security, {CCS} 2025, Taipei, Taiwan, October 13-17,
                  2025},
  pages        = {4184--4198},
  publisher    = {{ACM}},
  year         = {2025}
}

@inproceedings{WangZCBKY25,
  author       = {Zhiqi Wang and
                  Chengyu Zhang and
                  Yuetian Chen and
                  Nathalie Baracaldo and
                  Swanand Ravindra Kadhe and
                  Lei Yu},
  title        = {Membership Inference Attacks as Privacy Tools: Reliability, Disparity
                  and Ensemble},
  booktitle    = {Proceedings of {ACM} Conference on Computer and
                  Communications Security ({CCS}), Taipei, Taiwan, October 13-17,
                  2025},
  pages        = {1724--1738},
  publisher    = {{ACM}},
  year         = {2025}
}

@inproceedings{zhang25soft,
  author       = {Kaiyuan Zhang and
                  Siyuan Cheng and
                  Hanxi Guo and
                  Yuetian Chen and
                  Zian Su and
                  Shengwei An and
                  Yuntao Du and
                  Charles Fleming and
                  Ashish Kundu and
                  Xiangyu Zhang and
                  Ninghui Li},
  title        = {{SOFT:} Selective Data Obfuscation for Protecting {LLM} Fine-tuning
                  against Membership Inference Attacks},
  booktitle    = {34th {USENIX} Security Symposium, {USENIX} Security 2025, Seattle,
                  WA, USA, August 13-15, 2025},
  pages        = {8135--8154},
  publisher    = {{USENIX} Association},
  year         = {2025}
}

@inproceedings{li25enhanced,
  author       = {Hao Li and
                  Zheng Li and
                  Siyuan Wu and
                  Yutong Ye and
                  Min Zhang and
                  Dengguo Feng and
                  Yang Zhang},
  title        = {Enhanced Label-Only Membership Inference Attacks with Fewer Queries},
  booktitle    = {34th {USENIX} Security Symposium, {USENIX} Security 2025, Seattle,
                  WA, USA, August 13-15, 2025},
  pages        = {5465--5483},
  publisher    = {{USENIX} Association},
  year         = {2025}
}

@inproceedings{he2025llms,
  author       = {Yu He and
                  Boheng Li and
                  Liu Liu and
                  Zhongjie Ba and
                  Wei Dong and
                  Yiming Li and
                  Zhan Qin and
                  Kui Ren and
                  Chun Chen},
  title        = {Towards Label-Only Membership Inference Attack against Pre-trained
                  Large Language Models},
  booktitle    = {34th {USENIX} Security Symposium, {USENIX} Security 2025, Seattle,
                  WA, USA, August 13-15, 2025},
  pages        = {1609--1628},
  publisher    = {{USENIX} Association},
  year         = {2025}
}

@inproceedings{hu2025vlms,
  author       = {Yuke Hu and
                  Zheng Li and
                  Zhihao Liu and
                  Yang Zhang and
                  Zhan Qin and
                  Kui Ren and
                  Chun Chen},
  title        = {Membership Inference Attacks Against Vision-Language Models},
  booktitle    = {34th {USENIX} Security Symposium, {USENIX} Security 2025, Seattle,
                  WA, USA, August 13-15, 2025},
  pages        = {1589--1608},
  publisher    = {{USENIX} Association},
  year         = {2025}
}

@inproceedings{chen2025method,
  author       = {Zitao Chen and
                  Karthik Pattabiraman},
  title        = {A Method to Facilitate Membership Inference Attacks in Deep Learning
                  Models},
  booktitle    = {32nd Annual Network and Distributed System Security Symposium, {NDSS}
                  2025, San Diego, California, USA, February 24-28, 2025},
  publisher    = {The Internet Society},
  year         = {2025}
}

@inproceedings{pang2025black,
  author       = {Yan Pang and
                  Tianhao Wang},
  title        = {Black-box Membership Inference Attacks against Fine-tuned Diffusion
                  Models},
  booktitle    = {32nd Annual Network and Distributed System Security Symposium, {NDSS}
                  2025, San Diego, California, USA, February 24-28, 2025},
  publisher    = {The Internet Society},
  year         = {2025},
}

@inproceedings{Peng2025diffence,
  author       = {Yuefeng Peng and
                  Ali Naseh and
                  Amir Houmansadr},
  title        = {Diffence: Fencing Membership Privacy With Diffusion Models},
  booktitle    = {32nd Annual Network and Distributed System Security Symposium, {NDSS}
                  2025, San Diego, California, USA, February 24-28, 2025},
  publisher    = {The Internet Society},
  year         = {2025}
}

@inproceedings{shang2025defend,
  author       = {Jing Shang and
                  Jian Wang and
                  Kailun Wang and
                  Jiqiang Liu and
                  Nan Jiang and
                  Md. Armanuzzaman and
                  Ziming Zhao},
  title        = {Defending Against Membership Inference Attacks on Iteratively Pruned
                  Deep Neural Networks},
  booktitle    = {32nd Annual Network and Distributed System Security Symposium, {NDSS}
                  2025, San Diego, California, USA, February 24-28, 2025},
  publisher    = {The Internet Society},
  year         = {2025},
}

@inproceedings{Wang2025rigging,
  author       = {Zihao Wang and
                  Rui Zhu and
                  Zhikun Zhang and
                  Haixu Tang and
                  XiaoFeng Wang},
  title        = {Rigging the Foundation: Manipulating Pre-training for Advanced Membership
                  Inference Attacks},
  booktitle    = {{IEEE} Symposium on Security and Privacy, {S\&P} 2025, San Francisco,
                  CA, USA, May 12-15, 2025},
  pages        = {2509--2526},
  publisher    = {{IEEE}},
  year         = {2025}
}

@inproceedings{wen2024membership,
  title={Membership inference attacks against in-context learning},
  author={Wen, Rui and Li, Zheng and Backes, Michael and Zhang, Yang},
  booktitle={Proceedings of the 2024 on ACM SIGSAC Conference on Computer and Communications Security},
  pages={3481--3495},
  year={2024}
}

@inproceedings{liu2024please,
  title={Please tell me more: Privacy impact of explainability through the lens of membership inference attack},
  author={Liu, Han and Wu, Yuhao and Yu, Zhiyuan and Zhang, Ning},
  booktitle={2024 IEEE Symposium on Security and Privacy (S\&P)},
  pages={4791--4809},
  year={2024},
  organization={IEEE}
}

@inproceedings{chen2021machine,
  title={When machine unlearning jeopardizes privacy},
  author={Chen, Min and Zhang, Zhikun and Wang, Tianhao and Backes, Michael and Humbert, Mathias and Zhang, Yang},
  booktitle={Proceedings of the 2021 ACM SIGSAC Conference on Computer and Communications Security},
  pages={896--911},
  year={2021}
}

@inproceedings{stevanoski2024querycheetah,
  title={QueryCheetah: Fast Automated Discovery of Attribute Inference Attacks Against Query-Based Systems},
  author={Stevanoski, Bozhidar and Cretu, Ana-Maria and de Montjoye, Yves-Alexandre},
  booktitle={Proceedings of the 2024 on ACM SIGSAC Conference on Computer and Communications Security},
  pages={3451--3465},
  year={2024}
}

@inproceedings{zhu2024unified,
  title={A unified membership inference method for visual self-supervised encoder via part-aware capability},
  author={Zhu, Jie and Zha, Jirong and Li, Ding and Wang, Leye},
  booktitle={Proceedings of the 2024 on ACM SIGSAC Conference on Computer and Communications Security},
  pages={1241--1255},
  year={2024}
}

@inproceedings{liu2021encodermi,
  title={Encodermi: Membership inference against pre-trained encoders in contrastive learning},
  author={Liu, Hongbin and Jia, Jinyuan and Qu, Wenjie and Gong, Neil Zhenqiang},
  booktitle={Proceedings of the 2021 ACM SIGSAC Conference on Computer and Communications Security},
  pages={2081--2095},
  year={2021}
}

@article{yan2023mtl,
  title={MTL-leak: Privacy risk assessment in multi-task learning},
  author={Yan, Hongyang and Yan, Anli and Hu, Li and Liang, Jiaming and Hu, Haibo},
  journal={IEEE Transactions on Dependable and Secure Computing},
  volume={21},
  number={1},
  pages={204--215},
  year={2023},
  publisher={IEEE}
}

@inproceedings{zhifei2017cvpr,
  title={Age progression/regression by conditional adversarial autoencoder},
  author={Zhang, Zhifei and Song, Yang and Qi, Hairong},
  booktitle={Proceedings of the IEEE Conference on Computer Vision and Pattern Recognition},
  pages={5810--5818},
  year={2017}
}

@inproceedings{he2016deep,
  title={Deep residual learning for image recognition},
  author={He, Kaiming and Zhang, Xiangyu and Ren, Shaoqing and Sun, Jian},
  booktitle={Proceedings of the IEEE Conference on Computer Vision and Pattern Recognition},
  pages={770--778},
  year={2016}
}

@inproceedings{liu2015faceattributes,
  title={Deep learning face attributes in the wild},
  author={Liu, Ziwei and Luo, Ping and Wang, Xiaogang and Tang, Xiaoou},
  booktitle={Proceedings of the IEEE International Conference on Computer Vision},
  pages={3730--3738},
  year={2015}
}

@inproceedings{karkkainenfairface,
  title={Fairface: Face attribute dataset for balanced race, gender, and age for bias measurement and mitigation},
  author={Karkkainen, Kimmo and Joo, Jungseock},
  booktitle={Proceedings of the IEEE/CVF Winter Conference on Applications of Computer Vision},
  pages={1548--1558},
  year={2021}
}

@inproceedings{dosovitskiy2021an,
  author       = {Alexey Dosovitskiy and
                  Lucas Beyer and
                  Alexander Kolesnikov and
                  Dirk Weissenborn and
                  Xiaohua Zhai and
                  Thomas Unterthiner and
                  Mostafa Dehghani and
                  Matthias Minderer and
                  Georg Heigold and
                  Sylvain Gelly and
                  Jakob Uszkoreit and
                  Neil Houlsby},
  title        = {An Image is Worth 16x16 Words: Transformers for Image Recognition
                  at Scale},
  booktitle    = {9th International Conference on Learning Representations},
  year         = {2021}
}

@inproceedings{hu2021loralowrankadaptationlarge,
  author       = {Edward J. Hu and
                  Yelong Shen and
                  Phillip Wallis and
                  Zeyuan Allen{-}Zhu and
                  Yuanzhi Li and
                  Shean Wang and
                  Lu Wang and
                  Weizhu Chen},
  title        = {LoRA: Low-Rank Adaptation of Large Language Models},
  booktitle    = {The Tenth International Conference on Learning Representations},
  year         = {2022},
}

@article{tschandl2018ham10000,
  title={The HAM10000 dataset, a large collection of multi-source dermatoscopic images of common pigmented skin lesions},
  author={Tschandl, Philipp and Rosendahl, Cliff and Kittler, Harald},
  journal={Scientific Data},
  volume={5},
  number={1},
  pages={1--9},
  year={2018},
  publisher={Nature Publishing Group}
}

@inproceedings{schler2006effects,
    title={Effects of age and gender on blogging.},
    author={Schler, Jonathan and Koppel, Moshe and Argamon, Shlomo and Pennebaker, James W},
    booktitle={AAAI Spring Symposium: Computational Approaches to Analyzing Weblogs},
    volume={6},
    pages={199--205},
    year={2006}
}

@article{qwen3technicalreport,
  title={Qwen3 technical report},
  author={Yang, An and Li, Anfeng and Yang, Baosong and Zhang, Beichen and Hui, Binyuan and Zheng, Bo and Yu, Bowen and Gao, Chang and Huang, Chengen and Lv, Chenxu and others},
  journal={arXiv preprint arXiv:2505.09388},
  year={2025}
}

@inproceedings{standley2020taskslearnedmultitasklearning,
  title={Which tasks should be learned together in multi-task learning?},
  author={Standley, Trevor and Zamir, Amir and Chen, Dawn and Guibas, Leonidas and Malik, Jitendra and Savarese, Silvio},
  booktitle={International Conference on Machine Learning},
  pages={9120--9132},
  year={2020},
  organization={PMLR}
}

@article{yu2020gradientsurgerymultitasklearning,
  title={Gradient surgery for multi-task learning},
  author={Yu, Tianhe and Kumar, Saurabh and Gupta, Abhishek and Levine, Sergey and Hausman, Karol and Finn, Chelsea},
  journal={Advances in Neural Information Processing Systems},
  volume={33},
  pages={5824--5836},
  year={2020}
}

@inproceedings{dwork2006calibrating,
  title={Calibrating noise to sensitivity in private data analysis},
  author={Dwork, Cynthia and McSherry, Frank and Nissim, Kobbi and Smith, Adam},
  booktitle={Theory of Cryptography Conference},
  pages={265--284},
  year={2006},
  organization={Springer}
}

@inproceedings{abadi2016deep,
  title={Deep learning with differential privacy},
  author={Abadi, Martin and Chu, Andy and Goodfellow, Ian and McMahan, H Brendan and Mironov, Ilya and Talwar, Kunal and Zhang, Li},
  booktitle={Proceedings of the 2016 ACM SIGSAC Conference on Computer and Communications Security},
  pages={308--318},
  year={2016}
}

@inproceedings{yousefpour2021opacus,
  title={Opacus: User-Friendly Differential Privacy Library in PyTorch},
  author={Yousefpour, Ashkan and Shilov, Igor and Sablayrolles, Alexandre and Testuggine, Davide and Prasad, Karthik and Malek, Mani and Nguyen, John and Ghosh, Sayan and Bharadwaj, Akash and Zhao, Jessica and others},
  booktitle={NeurIPS 2021 Workshop Privacy in Machine Learning},
  year={2021}
}

@inproceedings{duan2024do,
  title={Do Membership Inference Attacks Work on Large Language Models?},
  author={Duan, Michael and Suri, Anshuman and Mireshghallah, Niloofar and Min, Sewon and Shi, Weijia and Zettlemoyer, Luke and Tsvetkov, Yulia and Choi, Yejin and Evans, David and Hajishirzi, Hannaneh},
  booktitle={First Conference on Language Modeling},
  year={2024}
}

@inproceedings{hayesexploring,
  title={Exploring the limits of strong membership inference attacks on large language models},
  author={Hayes, Jamie and Shumailov, Ilia and Choquette-Choo, Christopher A and Jagielski, Matthew and Kaissis, Georgios and Nasr, Milad and Annamalai, Meenatchi Sundaram Muthu Selva and Mireshghallah, Niloofar and Shilov, Igor and Meeus, Matthieu and others},
  booktitle={The Thirty-ninth Annual Conference on Neural Information Processing Systems},
  year={2025}
}

@article{dwork2014algorithmic,
  title={The algorithmic foundations of differential privacy},
  author={Dwork, Cynthia and Roth, Aaron},
  journal={Foundations and Trends{\textregistered} in Theoretical Computer Science},
  volume={9},
  number={3--4},
  pages={211--407},
  year={2014},
  publisher={Now Publishers, Inc.}
}

@inproceedings{lin2014microsoft,
  title={Microsoft coco: Common objects in context},
  author={Lin, Tsung-Yi and Maire, Michael and Belongie, Serge and Hays, James and Perona, Pietro and Ramanan, Deva and Doll{\'a}r, Piotr and Zitnick, C Lawrence},
  booktitle={European Conference on Computer Vision},
  pages={740--755},
  year={2014},
  organization={Springer}
}

@inproceedings{nasr2021adversaryinstantiationlowerbounds,
  title={Adversary instantiation: Lower bounds for differentially private machine learning},
  author={Nasr, Milad and Songi, Shuang and Thakurta, Abhradeep and Papernot, Nicolas and Carlin, Nicholas},
  booktitle={2021 IEEE Symposium on Security and Privacy (S\&P)},
  pages={866--882},
  year={2021},
  organization={IEEE}
}

@inproceedings{choquette2021label,
  title={Label-only membership inference attacks},
  author={Choquette-Choo, Christopher A and Tramer, Florian and Carlini, Nicholas and Papernot, Nicolas},
  booktitle={International Conference on Machine Learning},
  pages={1964--1974},
  year={2021},
  organization={PMLR}
}

@inproceedings{ganta2008composition,
  title={Composition attacks and auxiliary information in data privacy},
  author={Ganta, Srivatsava Ranjit and Kasiviswanathan, Shiva Prasad and Smith, Adam},
  booktitle={Proceedings of the 14th ACM SIGKDD International Conference on Knowledge Discovery and Data Mining},
  pages={265--273},
  year={2008}
}

@inproceedings{dinur2003revealing,
  title={Revealing information while preserving privacy},
  author={Dinur, Irit and Nissim, Kobbi},
  booktitle={Proceedings of the Twenty-second ACM SIGMOD-SIGACT-SIGART Symposium on Principles of Database Systems},
  pages={202--210},
  year={2003}
}

@misc{roberta_squad2_deepset,
  howpublished = {\url{https://huggingface.co/deepset/roberta-base-squad2}}
}

@misc{t5_qa_qg_valhalla,
  howpublished = {\url{https://huggingface.co/valhalla/t5-base-qa-qg-hl}}
}

@inproceedings{rajpurkar2016squad,
  title={Squad: 100,000+ questions for machine comprehension of text},
  author={Rajpurkar, Pranav and Zhang, Jian and Lopyrev, Konstantin and Liang, Percy},
  booktitle={Proceedings of the 2016 Conference on Empirical Methods in Natural Language Processing},
  pages={2383--2392},
  year={2016}
}

@inproceedings{tong2026membership,
  title={Membership Inference Attacks on Tokenizers of Large Language Models},
  author={Tong, Meng and Du, Yuntao and Chen, Kejiang and Zhang, Weiming and Li, Ninghui},
  booktitle = {35th USENIX Security Symposium (USENIX Security 26)}, 
  year = {2026}, 
  publisher = {USENIX Association}
}

@inproceedings{du2025imitative,
  title={Imitative Membership Inference Attack},
  author={Du, Yuntao and Chen, Yuetian and Xiao, Hanshen and Ribeiro, Bruno and Li, Ninghui},
  booktitle = {35th USENIX Security Symposium (USENIX Security 26)}, 
  year = {2026}, 
  publisher = {USENIX Association}
}

@inproceedings{chen2026window,
  title={Window-based Membership Inference Attacks Against Fine-tuned Large Language Models},
  author={Chen, Yuetian and Du, Yuntao and Zhang, Kaiyuan and Kundu, Ashish and Fleming, Charles and Ribeiro, Bruno and Li, Ninghui},
  booktitle = {35th USENIX Security Symposium (USENIX Security 26)}, 
  year = {2026}, 
  publisher = {USENIX Association}
}

@inproceedings{wang2026inference,
  title={Inference Attacks Against Graph Generative Diffusion Models},
  author={Wang, Xiuling and Huang, Xin and Luo, Guibo and Xu, Jianliang},
  booktitle = {35th USENIX Security Symposium (USENIX Security 26)}, 
  year = {2026}, 
  publisher = {USENIX Association}
}

@inproceedings{li2025compleak,
  title={CompLeak: Deep Learning Model Compression Exacerbates Privacy Leakage},
  author={Li, Na and Gao, Yansong and Hu, Hongsheng and Kuang, Boyu and Fu, Anmin},
  booktitle = {35th USENIX Security Symposium (USENIX Security 26)}, 
  year = {2026}, 
  publisher = {USENIX Association}
}

@inproceedings{wang2026vidleaks,
  title={VidLeaks: Membership Inference Attacks Against Text-to-Video Models},
  author={Wang, Li and Chen, Wenyu and Yu, Ning and Li, Zheng and Guo, Shanqing},
  booktitle = {35th USENIX Security Symposium (USENIX Security 26)}, 
  year = {2026}, 
  publisher = {USENIX Association}
}

%%%%%%%%%%%%%%%%%%%%%%%%%%%%%%%%%%%%%%%%%%%%%%%%%%%%%%%%%%%%%%%%%%%%%%%%%%%%%%%%
\appendix

\subsection{Single-Task MIA Performance and Victim Model Utility on CelebA}\label{app:single_task_celeba}

Table~\ref{tab:single_task_celeba} reports the per-task MIA performance and the underlying victim model utility for the seven CelebA tasks used in the divergence study (Section~\ref{sec:divergence}). All models are ResNet-18 independently trained for 60 epochs with identical hyperparameters. The training accuracy approaches 100\% for all tasks, indicating that each victim model fully fits its training set. However, the gap between training and validation accuracy varies notably across tasks: tasks with higher validation utility (e.g., Hat at 98.72\%, Male at 97.90\%) exhibit smaller train-val gaps and tend to expose weaker single-task MIA signals (AUC $\approx 0.51$--$0.53$), while tasks with lower validation utility (e.g., SHair at 72.71\%, Nose at 80.10\%) exhibit larger generalization gaps and yield stronger single-task leakage (AUC $\approx 0.64$). This is consistent with the well-known correlation between overfitting and membership leakage. Despite the heterogeneous per-task leakage strengths, \name consistently amplifies privacy risk through joint fusion, as analyzed in Section~\ref{sec:divergence}.

\begin{table}[h]
\centering
\caption{Single-task MIA performance and victim model task utility on CelebA (60 epochs). Train Acc and Val Acc denote the victim model's accuracy on its training and validation sets, respectively, for its own classification task.}
\label{tab:single_task_celeba}
\resizebox{0.95\linewidth}{!}{
\begin{tabular}{lcccc}
\toprule
\textbf{Task} & \textbf{MIA AUC} & \textbf{MIA Acc.} & \textbf{Model Train Acc.} & \textbf{Model Val Acc.} \\
\midrule
Male  & 0.5309 & 0.5166 & 100.00\% & 97.90\% \\
Mouth & 0.5489 & 0.5264 & 99.91\%  & 92.64\% \\
Nose  & 0.6422 & 0.5842 & 99.68\%  & 80.10\% \\
BHair & 0.5981 & 0.5570 & 99.72\%  & 88.13\% \\
Smile & 0.5581 & 0.5326 & 99.71\%  & 91.77\% \\
SHair & 0.6395 & 0.5834 & 99.66\%  & 72.71\% \\
Hat   & 0.5130 & 0.5088 & 100.00\% & 98.72\% \\
\bottomrule
\end{tabular}
}
\end{table}

\subsection{Prompt Template}\label{app:prompt_template}
% ======================================================

The Prompt templates used are shown in Table~\ref{tab:prompt_template}.

\begin{table}[h]
\centering
\caption{Prompt templates and example responses for LLM fine-tuning tasks. The model generates the response as free-form text following the prompt.}
\label{tab:prompt_template}
\resizebox{\linewidth}{!}{
\begin{tabular}{p{1.5cm}|p{5.0cm}|p{1.5cm}}
\toprule
\textbf{Task} & \textbf{Prompt Template} & \textbf{Response} \\
\midrule
Gender & Based on the following blog post, predict the author's gender (male or female).\newline\newline Blog: \textit{\{text\}}\newline\newline Gender: & male \\
\midrule
Horoscope & Based on the following blog post, predict the author's zodiac sign.\newline\newline Blog: \textit{\{text\}}\newline\newline Zodiac sign: & Gemini \\
\midrule
Age & Based on the following blog post, predict the author's age (as a number).\newline\newline Blog: \textit{\{text\}}\newline\newline Age: & 27 \\
\bottomrule
\end{tabular}
}
\end{table}

\subsection{Differential Privacy as Defense}\label{app:dpDefense}

\subsubsection{Experimental Setup}
We conduct DP defense experiments on UTKFace using a compact CNN architecture compatible with Opacus~\cite{yousefpour2021opacus}. Following the DP-SGD framework, we apply per-sample gradient clipping with maximum norm $C=1.0$ and add calibrated Gaussian noise during training. We evaluate two noise multiplier configurations: $\sigma=0.2$ (moderate noise) and $\sigma=0.5$ (strong noise), compared against the non-private baseline. All models are trained for 50 epochs with privacy parameter $\delta=10^{-5}$.

To simulate a realistic attack scenario, we train \textit{only the victim models} with DP-SGD while keeping shadow models non-private. This represents a conservative evaluation where the adversary cannot replicate the victim's privacy-preserving training procedure.

\subsubsection{Result and Analysis}

\mypara{Impact on Model Utility.}
Table~\ref{tab:dp_utility} presents the task performance under different DP configurations. As expected, adding noise during training degrades model utility across all tasks. With $\sigma=0.2$, the Race classification accuracy drops by 4.6\% (from 76.58\% to 71.98\%), while Age MAE increases by 1.25 years. The stronger $\sigma=0.5$ configuration incurs larger utility loss: Race accuracy decreases by 9.58\% and Age MAE increases by 2.73 years. Gender classification proves more robust, with accuracy remaining above 86\% under strong DP.

\mypara{Defense Effectiveness Against MIA.}
Table~\ref{tab:dp_attack} shows the MIA performance under DP protection. The results demonstrate that DP-SGD provides substantial defense against both single-task and joint attacks.

Without privacy protection, \name achieves an AUC of 0.879, substantially threatening membership privacy. With moderate DP ($\sigma=0.2$), the joint attack AUC drops dramatically to 0.550, only marginally above random guessing. Strong DP ($\sigma=0.5$) further reduces the AUC to 0.519, effectively neutralizing the attack. Single-task attacks are similarly mitigated, with average AUC decreasing from 0.674 to approximately 0.51 to 0.53.

% To quantify the defense effectiveness, we define the \textit{Protection Rate} as the fraction of excess attack AUC (above the random baseline of 0.5) that is eliminated:
% \begin{equation}
% \text{Protection Rate} = \frac{\text{AUC}_{\text{no-DP}} - \text{AUC}_{\text{DP}}}{\text{AUC}_{\text{no-DP}} - 0.5} \times 100\%.
% \end{equation}

\mypara{Privacy-Utility Trade-off Analysis.} Results in Table~\ref{tab:dp_tradeoff} reveal a favorable privacy-utility trade-off. Moderate DP ($\sigma=0.2$) reduces \name attack AUC by 0.329 (from 0.879 to 0.550), achieving 86.7\% protection rate while incurring only 4.60\% Race classification accuracy loss. Strong DP ($\sigma=0.5$) provides even stronger defense with 95.1\% protection rate, though at the cost of 9.58\% accuracy degradation. For privacy-sensitive applications such as medical imaging, sacrificing moderate utility to achieve near-complete attack mitigation represents an acceptable trade-off.

\mypara{Theoretical vs. Empirical Privacy.}
A noteworthy observation is the gap between theoretical privacy guarantees and empirical attack effectiveness. The $\sigma=0.2$ configuration yields $\varepsilon \approx 635$, which theoretically provides weak privacy guarantees according to the standard DP framework. However, the empirical \name attack AUC of 0.550 indicates strong practical protection, nearly reducing the attack to random guessing. Conversely, $\sigma=0.5$ achieves a much tighter $\varepsilon \approx 22.7$, yet the empirical improvement is marginal (AUC: 0.550 $\rightarrow$ 0.519). This discrepancy aligns with recent findings that DP-SGD's empirical privacy often exceeds its theoretical bounds~\cite{nasr2021adversaryinstantiationlowerbounds}, suggesting that $\varepsilon$ alone may not fully capture real-world privacy risks against MIA.

\mypara{Residual Privacy Amplification Under DP.}
Despite the substantial protection offered by DP, we observe that \name \textit{consistently outperforms} single-task MIAs even under strong privacy protection. As shown in Table~\ref{tab:dp_attack}, with $\sigma=0.5$, the joint attack achieves AUC of 0.519 compared to the average single-task AUC of 0.511. While both values are close to random guessing, the persistent gap indicates that the privacy amplification effect from ODMM exposure is not entirely eliminated by DP. This residual risk, though small, underscores that multi-model deployments require more conservative privacy budgets than single-model scenarios.

% \mypara{Practical Deployment Considerations.}
% While DP effectively mitigates \name, several considerations remain for real-world ODMM deployments:

% \noindent$\bullet$\textbf{Uniform Protection Requirement}: All ODMM models sharing the same training data must be protected with DP. A single unprotected model can leak membership signals that attackers exploit, undermining the privacy guarantees of other protected models.

% \noindent$\bullet$\textbf{Privacy Budget Allocation}: When deploying $K$ ODMM models, the total privacy leakage compounds across models. Organizations should account for this cumulative exposure when setting per-model privacy budgets.

% \noindent$\bullet$\textbf{Utility-Sensitive Domains}: In applications where utility degradation is unacceptable (e.g., safety critical medical diagnosis), alternative defenses such as output perturbation~\cite{dwork2014algorithmic} or access control mechanisms may complement or substitute DP-SGD.

\subsection{Hard-Label Experiment Details}
\label{app:hard_label}

\mypara{Perturbation Setup.}
For the hard-label experiments described in Section~\ref{sec:hard_label}, we utilize a noise-based estimation method. For a given input image $x$, we generate $N=30$ perturbed samples $x_i = x + \delta_i$, where $\delta_i \sim \mathcal{N}(0, \sigma^2)$ with $\sigma=0.05$. We query each task-specific model with these perturbed samples to obtain a set of labels.

\mypara{Feature Extraction.}
Since the models do not return probabilities, we construct features based on label stability. For classification tasks (Gender, Race), we directly compare the predicted class indices. For the regression task (Age), we discretize the output into 12 bins (0-9, 10-19, etc.) to treat it as a classification problem for consistency calculation. The extracted features for each task are:
\begin{itemize}[leftmargin=*]
    \item \textbf{Original Correctness}: Boolean indicating if the prediction on the clean image matches the ground truth.
    \item \textbf{Consistency Ratio}: The fraction of perturbed samples $x_i$ that yield the same prediction as the clean image $x$.
    \item \textbf{Robustness Accuracy}: The fraction of perturbed samples $x_i$ that are classified correctly.
    \item \textbf{Label Entropy}: The entropy of the distribution of predicted labels across the $N$ perturbations.
\end{itemize}

\mypara{Joint Attack Construction.}
The \name attack model aggregates the extracted 4-dimensional feature vectors from all three tasks, resulting in a 12-dimensional joint feature vector. A Random Forest classifier (100 estimators) is trained on this joint representation to infer membership.

\subsection{Proof of Theorem~\ref{thm:monotone}}
\label{app:monotone}
The proof follows from a data processing argument in reverse. Let $\sigma(O_{1:k})$ denote the $\sigma$-algebra generated by the first $k$ model outputs and $\sigma(O_{1:k+1})$ the $\sigma$-algebra generated by $k+1$ outputs.
Since $O_{1:k}$ is a deterministic function of $O_{1:k+1}$ (projection onto the first $k$ coordinates), we have $\sigma(O_{1:k}) \subseteq \sigma(O_{1:k+1})$.
Moving from $k$ to $k+1$ models enlarges the observable $\sigma$-algebra.
 
The Bayes-optimal membership inference attack corresponds to the likelihood-ratio test on the observed data. For any test $\mathcal{A}_k^*$ that is measurable with respect to $\sigma(O_{1:k})$, we can construct the same test as a function of $O_{1:k+1}$ by simply ignoring the $(k+1)$-th output. Therefore, the supremum over all tests measurable with respect to $\sigma(O_{1:k+1})$ is taken over a strictly larger class of functions, yielding:
\[
    \scalebox{0.85}{$
    \mathrm{Adv}(\mathcal{A}_{k+1}^*) = \sup_{\mathcal{A} \in \sigma(O_{1:k+1})} \mathrm{Adv}(\mathcal{A}) \;\geq\; \sup_{\mathcal{A} \in \sigma(O_{1:k})} \mathrm{Adv}(\mathcal{A}) = \mathrm{Adv}(\mathcal{A}_{k}^*)
    $}
\]

\subsection{Proof of Theorem~\ref{thm:llr-decomp}}
\label{app:llr-decomp}
Given the leakage $S_1(z), \ldots, S_k(z)$ from $k$ ODMM models, 
the joint log-likelihood ratio is defined as:
\[
    \Lambda_k(z) = \log 
    \frac{p(S_1, \ldots, S_k \mid M=1)}
    {p(S_1, \ldots, S_k \mid M=0)},
\]
where $M \in \{0,1\}$ denotes the membership 
status and $p(\cdot \mid M)$ denotes the joint 
conditional distribution of all leakage 
given $M$. Under the conditional independence 
assumption (Assumption~\ref{asmp:cond-indep}), 
the joint conditional distribution factorizes 
into a product of per-model terms:
\[
    p(S_1, \ldots, S_k \mid M) = 
    \prod_{i=1}^{k} p_i(S_i \mid M).
\]
Substituting this factorization into the 
log-likelihood expression and applying the 
logarithm of a product yields:
\[
    \Lambda_k(z) = \log \prod_{i=1}^{k} 
    \frac{p_i(S_i(z) \mid M=1)}
    {p_i(S_i(z) \mid M=0)} 
    = \sum_{i=1}^{k} \log 
    \frac{p_i(S_i(z) \mid z \in \mathcal{D})}
    {p_i(S_i(z) \mid z \notin \mathcal{D})},
\]
which is the additive decomposition stated in Eq.~\eqref{eq:llr}. 
%Each term in the summation represents the membership evidence contributed by the $i$-th model independently.

\subsection{Proof of Theorem~\ref{thm:keff}}
\label{app:keff}
For each model $f_i$, let $w_i$ denote the standalone leakage strength, defined as the expected separation of its leakage between members and non-members:
\[
    w_i \;=\; \mathbb{E}[S_i \mid M=1] 
    \;-\; \mathbb{E}[S_i \mid M=0].
\]
The meta-classifier of \name aggregates the leakage via concatenation, producing a joint statistic whose overall leakage is:
\[
    \mathbb{E}\Bigl[\sum_{i=1}^{k} S_i 
    \mid M=1\Bigr] - 
    \mathbb{E}\Bigl[\sum_{i=1}^{k} S_i 
    \mid M=0\Bigr] \;=\; 
    \sum_{i=1}^{k} w_i.
\]
Since the per-model leakage may be correlated with pairwise correlation $\rho_{ij} = \mathrm{Corr}(S_i, S_j)$, the variance of the aggregate statistic is:
\[
    \mathrm{Var}\Bigl(\sum_{i=1}^{k} 
    S_i\Bigr) \;=\; 
    \sum_{i=1}^{k}\sum_{j=1}^{k} 
    \sigma_i\, \sigma_j\, \rho_{ij},
\]
where $\sigma_i^2 = \mathrm{Var}(S_i)$. Setting 
$w_i = \sigma_i$ without loss of generality (by 
rescaling), the signal-to-noise ratio (SNR) of the 
aggregate statistic becomes:
\[
    \mathrm{SNR}_k \;=\; 
    \frac{\bigl(\sum_{i=1}^{k} 
    w_i\bigr)^2}{\sum_{i=1}^{k}\sum_{j=1}^{k} 
    w_i\, w_j\, \rho_{ij}} \;=\; 
    \alpha.
\]
Since the discriminative power of the aggregate statistic is determined by $\mathrm{SNR}_k$, the overall amplification is directly proportional to $\alpha$. 
%The boundary cases follow immediately: when $\rho_{ij} = 0$ for $i \neq j$, $\alpha = (\sum_i w_i)^2 / \sum_i w_i^2 \leq k$ with equality when all $w_i$ are equal; when $\rho_{ij} = 1$ for all pairs, $\alpha = 1$.

\end{document}